\documentclass{article}
\usepackage{bm}

\usepackage{color}
\usepackage{amsmath}
\usepackage{amssymb}\usepackage{amsthm}
\usepackage{amsfonts}
\usepackage{amscd}
\begin{document}
\theoremstyle{plain}
\newtheorem*{ithm}{Theorem}
\newtheorem*{idefn}{Definition}
\newtheorem{thm}{Theorem}[section]
\newtheorem{lem}[thm]{Lemma}
\newtheorem{dlem}[thm]{Lemma/Definition}
\newtheorem{prop}[thm]{Proposition}
\newtheorem{set}[thm]{Setting}
\newtheorem{cor}[thm]{Corollary}
\newtheorem*{icor}{Corollary}
\theoremstyle{definition}
\newtheorem{assum}[thm]{Assumption}
\newtheorem{notation}[thm]{Notation}
\newtheorem{defn}[thm]{Definition}
\newtheorem{clm}[thm]{Claim}
\newtheorem{ex}[thm]{Example}
\theoremstyle{remark}
\newtheorem{rem}[thm]{Remark}
\newcommand{\unit}{\mathbb I}
\newcommand{\ali}[1]{{\mathfrak A}_{[ #1 ,\infty)}}
\newcommand{\alm}[1]{{\mathfrak A}_{(-\infty, #1 ]}}
\newcommand{\nn}[1]{\lV #1 \rV}
\newcommand{\br}{{\mathbb R}}
\newcommand{\dm}{{\rm dom}\mu}
\newcommand{\lb}{l_{\bb}(n,n_0,k_R,k_L,\lal,\bbD,\bbG,Y)}
\newcommand{\Ad}{\mathop{\mathrm{Ad}}\nolimits}
\newcommand{\Proj}{\mathop{\mathrm{Proj}}\nolimits}
\newcommand{\RRe}{\mathop{\mathrm{Re}}\nolimits}
\newcommand{\RIm}{\mathop{\mathrm{Im}}\nolimits}
\newcommand{\Wo}{\mathop{\mathrm{Wo}}\nolimits}
\newcommand{\Prim}{\mathop{\mathrm{Prim}_1}\nolimits}
\newcommand{\Primz}{\mathop{\mathrm{Prim}}\nolimits}
\newcommand{\ClassA}{\mathop{\mathrm{ClassA}}\nolimits}
\newcommand{\Class}{\mathop{\mathrm{Class}}\nolimits}
\newcommand{\diam}{\mathop{\mathrm{diam}}\nolimits}
\def\qed{{\unskip\nobreak\hfil\penalty50
\hskip2em\hbox{}\nobreak\hfil$\square$
\parfillskip=0pt \finalhyphendemerits=0\par}\medskip}
\def\proof{\trivlist \item[\hskip \labelsep{\bf Proof.\ }]}
\def\endproof{\null\hfill\qed\endtrivlist\noindent}
\def\proofof[#1]{\trivlist \item[\hskip \labelsep{\bf Proof of #1.\ }]}
\def\endproofof{\null\hfill\qed\endtrivlist\noindent}

\newcommand{\varphii}{\varphi}
\newcommand{\pgs}{\caP_{\sigma}}
\newcommand{\oo}{{\boldsymbol\varphii}}
\newcommand{\caA}{{\mathcal A}}
\newcommand{\caB}{{\mathcal B}}
\newcommand{\caC}{{\mathcal C}}
\newcommand{\caD}{{\mathcal D}}
\newcommand{\caE}{{\mathcal E}}
\newcommand{\caF}{{\mathcal F}}
\newcommand{\caG}{{\mathcal G}}
\newcommand{\caH}{{\mathcal H}}
\newcommand{\caI}{{\mathcal I}}
\newcommand{\caJ}{{\mathcal J}}
\newcommand{\caK}{{\mathcal K}}
\newcommand{\caL}{{\mathcal L}}
\newcommand{\caM}{{\mathcal M}}
\newcommand{\caN}{{\mathcal N}}
\newcommand{\caO}{{\mathcal O}}
\newcommand{\caP}{{\mathcal P}}
\newcommand{\caQ}{{\mathcal Q}}
\newcommand{\caR}{{\mathcal R}}
\newcommand{\caS}{{\mathcal S}}
\newcommand{\caT}{{\mathcal T}}
\newcommand{\caU}{{\mathcal U}}
\newcommand{\caV}{{\mathcal V}}
\newcommand{\caW}{{\mathcal W}}
\newcommand{\caX}{{\mathcal X}}
\newcommand{\caY}{{\mathcal Y}}
\newcommand{\caZ}{{\mathcal Z}}
\newcommand{\bbA}{{\mathbb A}}
\newcommand{\bbB}{{\mathbb B}}
\newcommand{\bbC}{{\mathbb C}}
\newcommand{\bbD}{{\mathbb D}}
\newcommand{\bbE}{{\mathbb E}}
\newcommand{\bbF}{{\mathbb F}}
\newcommand{\bbG}{{\mathbb G}}
\newcommand{\bbH}{{\mathbb H}}
\newcommand{\bbI}{{\mathbb I}}
\newcommand{\bbJ}{{\mathbb J}}
\newcommand{\bbK}{{\mathbb K}}
\newcommand{\bbL}{{\mathbb L}}
\newcommand{\bbM}{{\mathbb M}}
\newcommand{\bbN}{{\mathbb N}}
\newcommand{\bbO}{{\mathbb O}}
\newcommand{\bbP}{{\mathbb P}}
\newcommand{\bbQ}{{\mathbb Q}}
\newcommand{\bbR}{{\mathbb R}}
\newcommand{\bbS}{{\mathbb S}}
\newcommand{\bbT}{{\mathbb T}}
\newcommand{\bbU}{{\mathbb U}}
\newcommand{\bbV}{{\mathbb V}}
\newcommand{\bbW}{{\mathbb W}}
\newcommand{\bbX}{{\mathbb X}}
\newcommand{\bbY}{{\mathbb Y}}
\newcommand{\bbZ}{{\mathbb Z}}
\newcommand{\mfp}{{\mathfrak p}}
\newcommand{\mfq}{{\mathfrak q}}
\newcommand{\str}{^*}
\newcommand{\lv}{\left \vert}
\newcommand{\rv}{\right \vert}
\newcommand{\lV}{\left \Vert}
\newcommand{\rV}{\right \Vert}
\newcommand{\la}{\left \langle}
\newcommand{\ra}{\right \rangle}
\newcommand{\ltm}{\left \{}
\newcommand{\rtm}{\right \}}
\newcommand{\lcm}{\left [}
\newcommand{\rcm}{\right ]}
\newcommand{\ket}[1]{\lv #1 \ra}
\newcommand{\bra}[1]{\la #1 \rv}
\newcommand{\kl}[1]{\bm {#1}}
\newcommand{\hln}[1]{\hat\Lambda_{#1}}
\newcommand{\lmk}{\left (}
\newcommand{\rmk}{\right )}
\newcommand{\al}{{\mathcal A}}
\newcommand{\md}{M_d({\mathbb C})}
\newcommand{\ainn}{\mathop{\mathrm{AInn}}\nolimits}
\newcommand{\id}{\mathop{\mathrm{id}}\nolimits}
\newcommand{\Tr}{\mathop{\mathrm{Tr}}\nolimits}
\newcommand{\Ran}{\mathop{\mathrm{Ran}}\nolimits}
\newcommand{\Ker}{\mathop{\mathrm{Ker}}\nolimits}
\newcommand{\Aut}{\mathop{\mathrm{Aut}}\nolimits}
\newcommand{\spn}{\mathop{\mathrm{span}}\nolimits}
\newcommand{\Mat}{\mathop{\mathrm{M}}\nolimits}
\newcommand{\UT}{\mathop{\mathrm{UT}}\nolimits}
\newcommand{\DT}{\mathop{\mathrm{DT}}\nolimits}
\newcommand{\GL}{\mathop{\mathrm{GL}}\nolimits}
\newcommand{\spa}{\mathop{\mathrm{span}}\nolimits}
\newcommand{\supp}{\mathop{\mathrm{supp}}\nolimits}
\newcommand{\rank}{\mathop{\mathrm{rank}}\nolimits}
\newcommand{\idd}{\mathop{\mathrm{id}}\nolimits}
\newcommand{\ran}{\mathop{\mathrm{Ran}}\nolimits}
\newcommand{\dr}{ \mathop{\mathrm{d}_{{\mathbb R}^k}}\nolimits} 
\newcommand{\dc}{ \mathop{\mathrm{d}_{\cc}}\nolimits} \newcommand{\drr}{ \mathop{\mathrm{d}_{\rr}}\nolimits} 
\newcommand{\zin}{\mathbb{Z}}
\newcommand{\rr}{\mathbb{R}}
\newcommand{\cc}{\mathbb{C}}
\newcommand{\ww}{\mathbb{W}}
\newcommand{\nan}{\mathbb{N}}\newcommand{\bb}{\mathbb{B}}
\newcommand{\aaa}{\mathbb{A}}\newcommand{\ee}{\mathbb{E}}
\newcommand{\pp}{\mathbb{P}}
\newcommand{\wks}{\mathop{\mathrm{wk^*-}}\nolimits}
\newcommand{\mk}{{\Mat_k}}
\newcommand{\mnz}{\Mat_{n_0}}
\newcommand{\mn}{\Mat_{n}}
\newcommand{\dist}{\dc}
\newcommand{\braket}[2]{\left\langle#1,#2\right\rangle}
\newcommand{\ketbra}[2]{\left\vert #1\right \rangle \left\langle #2\right\vert}
\newcommand{\abs}[1]{\left\vert#1\right\vert}
\newtheorem{nota}{Notation}[section]
\def\qed{{\unskip\nobreak\hfil\penalty50
\hskip2em\hbox{}\nobreak\hfil$\square$
\parfillskip=0pt \finalhyphendemerits=0\par}\medskip}
\def\proof{\trivlist \item[\hskip \labelsep{\bf Proof.\ }]}
\def\endproof{\null\hfill\qed\endtrivlist\noindent}
\def\proofof[#1]{\trivlist \item[\hskip \labelsep{\bf Proof of #1.\ }]}
\def\endproofof{\null\hfill\qed\endtrivlist\noindent}
\newcommand{\spp}{ \mathfrak S_{\mfp_1,\mfp_2,+1}^{(N)} }
\newcommand{\sppe}{\bbE\lmk  \mathfrak S_{\mfp_1,\mfp_2,+1}^{(N)} \rmk}

\newcommand{\Ln}[1]{\Lambda_{#1}}
\newcommand{\ZZ}{\bbZ_2\times\bbZ_2}
\newcommand{\SSS}{\mathcal{S}}
\newcommand{\cs}{S}
\newcommand{\ct}{t}
\newcommand{\hS}{S}
\newcommand{\vv}{{\boldsymbol v}}
\newcommand{\ala}{a}
\newcommand{\bet}{b}
\newcommand{\gam}{c}
\newcommand{\alphas}{\alpha}
\newcommand{\alphai}{\alpha^{(\sigma_{1})}}
\newcommand{\alphan}{\alpha^{(\sigma_{2})}}
\newcommand{\betas}{\beta}
\newcommand{\betai}{\beta^{(\sigma_{1})}}
\newcommand{\betan}{\beta^{(\sigma_{2})}}
\newcommand{\alphass}{\alpha^{{(\sigma)}}}
\newcommand{\uu}{V}
\newcommand{\vp}{\varsigma}
\newcommand{\vpr}{R}
\newcommand{\tg}{\tau_{\Gamma}}
\newcommand{\sgg}{\Sigma_{\Gamma}^{(\sigma)}}
\newcommand{\nh}{t28}
\newcommand{\rk}{6}
\newcommand{\nii}{2}
\newcommand{\nhh}{28}
\newcommand{\sjt}{30}
\newcommand{\sjtg}{30}
\newcommand{\bcg}{\caB(\caH_{\alpha})\otimes  C^{*}(\Sigma)}
\title{Type of local von Neumann algebras in abelian quantum double model}
\author{Yoshiko Ogata}
\maketitle
\begin{abstract}
We show that the local von Neumann algebras on convex areas of the frustration-free ground state
of abelian quantum double models are of type $II_\infty$.
\end{abstract}
\section{Introduction}\label{intro}
The quantum double models were introduced by Kitaev in \cite{kitaev2003fault}.
They have substantial importance as the examples of topological phases.
An astonishing property of these models is the existence of quasi-particles, called anyons. 
Such anyons can be mathematically formulated as superselection sectors
\cite{N2}\cite{FN}\cite{CNN2}\cite{NaOg}\cite{ogata2022derivation}.
In this sense, the mathematical structures of
topological phases looks quite similar to that of algebraic quantum field theory (AQFT)
\cite{BDMRS}\cite{DHRI}\cite{BF}.
In spite of such similarity, there is some significant difference.
In this paper, we consider the local von Neumann algebra on convex cone areas of the frustration-free ground state of abelian quantum double models.
We show they are of type $II_\infty$.
This is in contrast to the situation of AQFT, where all
the local algebras are of type $III$.

Let us introduce the quantum double model. See \cite{FN} for a more detailed description.
Let $G$ be a finite group.
We denote by $\hat \bbZ^2$ the set of all edges of $\bbZ^2$.
All the horizontal edges have an orientation from left to right and 
all the vertical edges have an orientation from down to up. 
We denote by $v_{e,-1}$, $v_{e,+1}$, the vertices corresponding to the 
origin and target of the edge $e$, with respect to this orientation.

By a path of edges, we mean a sequence
of edges associated with direction
\[
{\mathfrak p}=(e_1,{\sigma_1})(e_2,{\sigma_2})\cdots(e_m,\sigma_m)
\]
such that 
\[
v_{e_{i}, \sigma_i}=v_{e_{i+1}, -\sigma_{i+1}},\quad i=1,\ldots, m-1.
\]
Here $e_i\in \hat \bbZ^2$ and  $\sigma_i=+1$ (resp. $\sigma_i=-1$) if, along the path,
we proceed in the direction (resp. opposite direction) of the given orientation of $e_i$.
We call $v_{e_1,-\sigma_1}$ the origin and $v_{e_m,\sigma_m}$ the target of $\mfp$.
The path $\mfp$ is self-avoiding if all of $v_{e_i,\pm 1}$ are different.
It is a loop if all the $v_{e_i,\pm 1}$ are different except for $v_{e_1,-\sigma_1}=v_{e_m,\sigma_m}$.

For a path $\mfp$, the symbol $\mfp^{-1}$ represents
the inverse path of $\mfp$.
If the target of $\mfp_1$ is the same as the origin of $\mfp_2$,
then $\mfp_1\mfp_2$ indicates the path obtained by connecting them.

For $\Lambda\subset \bbZ^2$, let $\caP_\Lambda$ be the set of squares included in $\Lambda$
.We set
$\tilde \Lambda:=\cup_{p\in \caP_\Lambda}\{ \text{edges in } p\}$, the set of all edges
forming a part of the squares in $\Lambda$.
The boundary of $\Lambda\subset \bbZ^2$ is 
\begin{align}
\partial\Lambda:=
\left\{
v\in \Lambda\mid 
\text{there exists } \;\;e\in \hat \bbZ^2 \;\;\text{such that}\;\;
\lmk v_{e,-1}=v,\text{and}\;v_{e,+1}\in \Lambda^c \rmk
\;\;\text{or}\;\;
\lmk v_{e,+1}=v,\text{and}\;v_{e,-1}\in \Lambda^c\rmk
\right\}.
\end{align}
For all $S\subset\bbZ^2$, $T\subset\hat\bbZ^2$,
$G^{S}$, $G^{T}$ denote the set of all the configulations
on $S$, $T$ respectively.
If $S_1\subset S_2$ and $\bm h\in G^{S_2}$, we set
$\bm h_{S_1}:=\bm h\vert_{S_1}$.

For each path $\mfp=(e_1,{\sigma_1})(e_2,{\sigma_2})\cdots(e_m,\sigma_m)$
where $e_i\in \bbZ^2$ with its direction $\sigma_i=\pm 1$
and $\bm h_{\mfp}=(h_e)_{e\in\mfp }\in G^\mfp$
we set
\begin{align}
\Psi_\mfp(\bm h_\mfp)=h_{e_1}^{\sigma_1}h_{e_2}^{\sigma_2}\cdots h_{e_m}^{\sigma_m}.
\end{align}

For each $(a,b)\in\bbZ^2$, we denote by $S_{(a,b)}$, the square
$(a,b)-(a+1,b)-(a+1,b+1)-(a,b+1)$.
We denote by $\mfp_{S_{(a,b)}}$, the path of edges
obtained by going around $S_{(a,b)}$ in the order
$(a,b)-(a+1,b)-(a+1,b+1)-(a,b+1)$.
We say a configuration $\bm g_{\hat \Lambda}\in G^{\hat\Lambda}$ on $\hat\Lambda\subset \hat\bbZ^2$
is admissible if for any square $S_{(a,b)}=(a,b)-(a+1,b)-(a+1,b+1)-(a,b+1)$ 
formed by $4$-edges in $\hat\Lambda$, we have
\[
g_{(a,b)-(a+1,b)}g_{(a+1,b)-(a+1,b+1)}=g_{(a,b)-(a,b+1)}g_{(a,b+1)-(a+1,b+1)}.
\]
The set of all admissible configurations  on $\hat \Lambda\subset \hat\bbZ^2$ is denoted by
$C_{\hat\Lambda}$.

For $\hat S\subset\hat\bbZ^2$, set
$\caB_{\hat S}:=\bigotimes_{\hat S} B(l^2(G))$.
Here $B(l^2(G))$ denotes the set of all bounded linear operators on $l^2(G)$
with CONS $\{\ket{g}\mid g\in G\}$.
For finite $S\subset\hat\bbZ^2$ and $\bm k\in G^S$,
we set $\ket {\bm k}:=\otimes_{e\in S}\ket{k_e}$.
The algebra, $\caB:=\caB_{\hat \bbZ^2}$
is our $2$-dimensional quantum spin system.
For $S_1\subset S_2$,
$\caB_{S_1}$ can be regarded as a sub $C^*$-algebra of $\caB_{S_2}$
naturally. The left regular representation of $G$ on $\bbC^{|G|}$ is denoted by $L_g$ and 
the right regular representation of $G$ on $\bbC^{|G|}$ is denoted by $R_g$:
\begin{align}
\begin{split}
L_g\ket{h}=\ket{gh},\quad
R_g\ket{h}=\ket{h g^{-1}},\quad g,h \in G.
\end{split}
\end{align}

For each vertex $v\in\bbZ^2$, we set
\begin{align}
A_v^{(g)}:=
\bigotimes_{e\in\hat \bbZ^2: v_{e,-1}=v}L_{g}^{e}
\bigotimes_{e\in \hat\bbZ^2: v_{e,+1}=v}R_{g^{-1}}^{e}.
\end{align}
We then set
\begin{align}
A_v:=\frac{1}{|G|}\sum_{g\in G}A_v^{(g)}.
\end{align}
For each square (plaquette) $p$, we set
\begin{align}
B_{p}:=\sum_{\bm k\in C_{\hat p}}\ket{\bm k}\bra{\bm k}.
\end{align}
From \cite{FN}, we know the following:
\begin{prop}[Proposition 2.1 \cite{FN} ]\label{fnu}
There exists a unique state $\omega_0$ on $\caB$ which satisfies
$\omega_0(A_v)=\omega_0(B_p)=1$
for all vertex $v$ and squares $p$ in $\bbZ^2$.
\end{prop}
Now we are ready to state our main theorem.
\begin{thm}\label{main}Suppose $G$ is abelian.
Let $(\caH_0,\pi_0,\Omega_0)$ be the GNS triple of $\omega_0$ in Proposition \ref{fnu}.
Let $\Gamma\subset \bbR^2$ be a convex cone.
Then the local von Neumann algebra
$\pi_0\lmk \caB_{\widetilde{\lmk \Gamma\cap \bbZ^2\rmk}}\rmk''$ is a type $II_\infty$ factor.
\end{thm}
From the general theory in Lemma 5.5 \cite{ogata2022derivation},
with the existence of nontrivial superselection sector and the Haag duality \cite{FN},
we know that $\pi_0\lmk \caB_{\widetilde{\lmk \Gamma\cap \bbZ^2\rmk}}\rmk''$  is
either type $II_\infty$ or type $III$ factor.
The question here is which of these occurs.

Theorem is restricted to abelian group $G$.
This is because
we use the result from \cite{FN},
the Haag duality and the existence of nontrivial superselection sector.
All other parts of the proof work for non-abelian cases.
Therefore, if the Haag duality and existence of nontrivial superselection sector
are proven for non-abelian quantum double models, our proof immediately shows 
that the local algebras for non-abelian quantum double models are type $II_\infty$ factors.

\section{The  frustration-free ground state $\omega_0$ of the quantum double model}\label{coffe}
In this section, we use the notation from section \ref{intro}, but we do not assume
$G$ to be abelian.
For each $N\in\bbN$ and $n_0,m_0\in \bbN$
let $\Lambda_N^{(n_0,m_0)}:=\lmk [-Nn_0,Nn_0]\times [-Nm_0,Nm_0]\rmk\cap \bbZ^2$
and $\hat\Lambda_N^{(n_0,m_0)}$ the set of all edges in $\Lambda_N^{(n_0,m_0)}$.
Let $\partial \hat \Lambda_N^{(n_0,m_0)}$ be the set of edges forming the boundary lines of 
$\Lambda_N^{(n_0,m_0)}=\lmk [-Nn_0,Nn_0]\times [-Nm_0,Nm_0]\rmk\cap \bbZ^2$.
In particular, we set $\Lambda_N:=\Lambda_N^{(1,1)}$
and $\hln N:=\hat\Lambda_N^{(1,1)}$, 
$\partial \hat \Lambda_N:=\partial \hat \Lambda_N^{(1,1)}$.
In this section, we derive a concrete expression of the restriction
of $\omega_0$ onto $\caB_{\hat\Lambda_N}$.

First we construct a state $\varphi_0$ by some concrete
formula and show that this $\varphi_0$ 
satisfies the condition $\varphi_0(A_v)=\varphi_0(B_p)=1$
for all vertex $v$ and squares $p$ in $\bbZ^2$.
From the uniqueness in Proposition \ref{fnu}, this means $\varphi_0=\omega_0$.
 For an event $\caS$, we set $\chi(\caS)$ to be $1$ if $\caS$ holds
 and  $0$ if $\caS$ does not hold.
 For a finite set $T$, $\lv T\rv$ denotes the cardinality of $T$.
\begin{lem}\label{pz}
There is a state $\varphi_0$ on $\caB_{\hat\bbZ^2}$
such that
\begin{align}\label{vpdef}
\begin{split}
\varphi_0\lmk
\ket{\bm h}_{\hat\Lambda_N}\bra{\bm k}
\rmk
=\frac{1}{|C_{\hat\Lambda_N}|}
\chi\lmk {\bm k, \bm h\in C_{\hat\Lambda_N}}\rmk
\chi\lmk \bm k_{\partial{\hln N}}=\bm h_{\partial{\hln N}}\rmk,
\end{split}
\end{align}
for all $N\in\bbN$, $\bm k, \bm h\in G^{\hat\Lambda_N}$.
\end{lem}
\begin{proof}
For each $N\in \bbN$ and $\bm a\in G^{\partial\hat \Lambda_N}$, set
\begin{align}
\caI_{\bm a}^{(N)}:=\left\{ \bm k \in  C_{\hln{N}}\mid \bm k_{\partial \hat \Lambda_N}=\bm a\right\}.
\end{align}
We denote by ${\mathbb A}_{\hat\Lambda_N}$
the set of all $\bm a\in G^{\partial\hat \Lambda_N}$
such that $\caI_{\bm a}^{(N)}\neq\emptyset$.
For $\bm a\in {\mathbb A}_{\hat\Lambda_N}$,
we set
\begin{align}
\psi_{\bm a}:= \sum_{\kl{k}\in \caI_{\bm a}^{(N)}}
\ket{\kl{k}}_{\hln N},\quad\quad
\hat\psi_{\bm a}:=\frac{1}{|\caI_{\bm a}^{(N)}|^{\frac 12}}\psi_{\bm a}.
\end{align}

Then
\begin{align}
\begin{split}
\varphi_{\hln N}(B):=\frac{1}{|C_{\hln{N}}|}
\sum_{\bm a\in {\mathbb A}_{\hat\Lambda_N}}
\braket{\psi_{\bm a}}{ B\psi_{\bm a}},\quad B\in\caB_{\hln N}
\end{split}
\end{align}
defines a state on $\caB_{\hln N}$.
For $\bm k, \bm h\in G^{\hat\Lambda_N}$,
we have
\begin{align}\label{neb}
\begin{split}
\varphi_{\hln {N}}\lmk
\ket{\bm h}_{\hat\Lambda_N}\bra{\bm k}
\rmk
=\frac{1}{|C_{\hat\Lambda_N}|}
\chi\lmk {\bm k, \bm h\in C_{\hat\Lambda_N}}\rmk
\chi\lmk \bm k_{\partial{\hln N}}=\bm h_{\partial{\hln N} }\rmk.
\end{split}
\end{align}

We have
\begin{align}\label{rest}
\varphi_{\hln {N+1}}\vert_{\caB_{\hln N}}=\varphi_{\hln N}.
\end{align}
In order to show this, note that for each
$\kl k \in C_{\hln N}$, 
there are $|G|^{4(2N+2)}$-number of $\bm m\in G^{\hln {N+1}\setminus \hln N}$
such that 
\begin{align}\label{kmkm}
\begin{split}
\lmk  \bm k,\bm m\rmk\in C_{\hat\Lambda_{N+1}}.
\end{split}
\end{align}
In fact, fix $\kl k \in C_{\hln N}$.
We freely fix $\bm m_i\in G^{L_i}$, $i=1,2,3,4$
with
\begin{align}
\begin{split}
&L_1:=\text{edges in}\;\; \lmk \lmk [-N,N+1]\times\{N+1\}\rmk \cup\lmk \{-N\}\times[N,N+1]\rmk\rmk,\\
&L_2:=\text{edges in}\;\; \lmk \lmk \{N+1\}\times [-N-1,N]\rmk\cup  \lmk [N,N+1]\times \{N\}\rmk\rmk\\
&L_3:=\text{edges in}\; \;\lmk \lmk [-N-1,N]\times\{-N-1\}\rmk\cup\lmk \{N\}\times[-N-1,-N]\rmk,\rmk\\
&L_4:=\text{edges in}\;\; \lmk \lmk \{-N-1\}\times [-N,N+1]\rmk \cup[-N-1, -N]\times \{-N\}\rmk
\end{split}
\end{align}
Then the configuration of all the other edges is determined uniquely
so that all
 the admissibility conditions are satisfied.
Hence the possible number of $\bm m\in G^{\hln {N+1}\setminus \hln N}$
satisfying (\ref{kmkm}) for a given $\bm k\in C_{\hln N}$
is
\begin{align}
\begin{split}
|G|^{|L_1|+|L_2|+|L_3|+|L_4|}
=|G|^{4(2N+2)}.
\end{split}
\end{align}
In particular we have
\begin{align}
|G|^{4(2N+2)}|C_{\hat\Lambda_{N}}|=|C_{\hat\Lambda_{N+1}}|.
\end{align}
Hence
 from (\ref{neb}) for $N+1$, we have
\begin{align}
\begin{split}
&\varphi_{\hln {N+1}}\lmk
\ket{\bm h}_{\hat\Lambda_N}\bra{\bm k}
\rmk\\
&=\sum_{\bm m\in G^{\hln {N+1}\setminus \hln N}}\varphi_{\hln {N+1}}\lmk
\ket{\bm h,\bm m}_{\hat\Lambda_{N+1}}\bra{\bm k,\bm m}
\rmk\\
&=\sum_{\bm m\in G^{\hln {N+1}\setminus \hln N}}\frac{1}{|C_{\hat\Lambda_{N+1}}|}
\chi\lmk {\lmk \bm k,\bm m\rmk, 
\lmk  \bm h,\bm m\rmk\in C_{\hat\Lambda_{N+1}}}\rmk
\chi\lmk\lmk \bm k,\bm m\rmk_{\partial{\hln {N+1}} }
= \lmk \bm h,\bm m\rmk_{\partial{\hln {N+1}}}\rmk\\
&=\sum_{\bm m\in G^{\hln {N+1}\setminus \hln N}}\frac{1}{|C_{\hat\Lambda_{N+1}}|}
\chi\lmk {\lmk \bm k,\bm m\rmk, 
\lmk  \bm h,\bm m\rmk\in C_{\hat\Lambda_{N+1}}}\rmk\\
&=
\chi\lmk { \bm k, 
  \bm h\in C_{\hat\Lambda_{N}}}\rmk
\chi\lmk\bm k_{\partial{\hln N}} =\bm h_{\partial{\hln N}} \rmk\sum_{\bm m\in G^{\hln {N+1}\setminus \hln N}}\frac{1}{|C_{\hat\Lambda_{N+1}}|}\chi\lmk {
\lmk  \bm h,\bm m\rmk\in C_{\hat\Lambda_{N+1}}}\rmk\\
&=\frac{1}{|C_{\hat\Lambda_{N+1}}|}\chi\lmk { \bm k, 
  \bm h\in C_{\hat\Lambda_{N}}}\rmk
\chi\lmk\bm k_{\partial{\hln N}} = \bm h_{\partial{\hln N}} \rmk
\lv
\left\{
\bm m\mid \lmk  \bm h,\bm m\rmk\in C_{\hat\Lambda_{N+1}}
\right
\}\rv\\
&=\frac{1}{|C_{\hat\Lambda_{N}}|}\chi\lmk { \bm k, 
  \bm h\in C_{\hat\Lambda_{N}}}\rmk
\chi\lmk\bm k_{\partial{\hln N}} = \bm h_{\partial{\hln N}} \rmk\\
&=\varphi_{\hln {N}}\lmk
\ket{\bm h}_{\hat\Lambda_N}\bra{\bm k}
\rmk.
\end{split}
\end{align}
In the third equality, we noted the fact that $\lmk \bm k,\bm m\rmk, 
\lmk  \bm h,\bm m\rmk\in C_{\hat\Lambda_{N+1}}$
trivially implies $\lmk \bm k,\bm m\rmk_{\partial{\hln {N+1}} }
= \lmk \bm h,\bm m\rmk_{\partial{\hln {N+1}}}$.
In the fourth equality we noted the fact that  $\lmk \bm k,\bm m\rmk, 
\lmk  \bm h,\bm m\rmk\in C_{\hat\Lambda_{N+1}}$
implies $\bm k_{\partial{\hln N}} =\bm h_{\partial{\hln N}}$,
because the labels of all the edges of $\bm k$ and $\bm h$ on ${\partial{\hln N}}$
are determined by $\bm m$ uniquely by the admissibility condition.
Furthermore, if $\lmk \bm h,\bm m\rmk\in C_{\hat\Lambda_{N+1}}$,
$\bm k_{\partial{\hln N}} =\bm h_{\partial{\hln N}}$,
$ \bm k, 
  \bm h\in C_{\hat\Lambda_{N+1}}$, then
 $\lmk \bm k,\bm m\rmk\in C_{\hat\Lambda_{N}}$
 holds because all the admissibility conditions of 
 $\lmk \bm k,\bm m\rmk\in C_{\hat\Lambda_{N+1}}$ which are not
 included in $\bm k\in C_{\hat\Lambda_{N}}$
 are the ones including only $\bm k_{\partial{\hln N}} =\bm h_{\partial{\hln N}}$
 and $\bm m$.
In the fifth equality we used the fact 
$\lv
\left\{
\bm m\mid \lmk  \bm h,\bm m\rmk\in C_{\hat\Lambda_{N+1}}
\right
\}\rv
=|G|^{4(2N+2)}=
\frac{|C_{\hat\Lambda_{N+1}}|}{|C_{\hat\Lambda_{N}}|}
$,
observed above.

This proves the claim (\ref{rest}).
Hence the consistency condition holds and we can extend the states $\varphi_{\hln N}$
to a state on $\caB$, obtaining 
 the desired state $\varphi_0$.
\end{proof}

\begin{lem}
For all $N\in\bbN$, $\bm k, \bm h\in G^{\hat\Lambda_N^{(n_0,m_0)}}$,
\begin{align}
\begin{split}
\varphi_0\lmk
\ket{\bm h}_{\hat\Lambda_N^{(n_0,m_0)}}\bra{\bm k}
\rmk
=\frac{1}{|C_{\hat\Lambda_N^{(n_0,m_0)}}|}
\chi\lmk {\bm k, \bm h\in C_{\hat\Lambda_N^{(n_0,m_0)}}}\rmk
\chi\lmk \bm k_{\partial{\hln N^{(n_0,m_0)}}}
=\bm h_{\partial{\hln N^{(n_0,m_0)}}}\rmk.
\end{split}
\end{align}
\end{lem}
\begin{proof}
We consider the case that $n_0<m_0$. The proof for $n_0>m_0$ is the same.
As in the proof of Theorem \ref{pz}, for any  $3\le N\in\bbN$ and 
$\bm h\in C_{\hat\Lambda_N^{(n_0,m_0)}}$, we have
\begin{align}\label{monb}
\begin{split}
 \lv
 \left\{
 \bm m\in G^{\hln {Nm_0}\setminus \hln N^{(n_0,m_0)}}\mid
 \lmk \bm h,\bm m\rmk \in C_{\hat\Lambda_{Nm_0}}
 \right\}
 \rv
 =|G|^{2\lmk 2m_0N+1\rmk\lmk m_0-n_0\rmk N}.
\end{split}
\end{align}
In particular, we have
\begin{align}\label{short}
\begin{split}
\lv C_{\hln N^{(n_0,m_0)}}\rv\cdot 
|G|^{2\lmk 2m_0N+1\rmk\lmk m_0-n_0\rmk N }=\lv C_{\hln {Nm_0}}\rv.
\end{split}
\end{align}

For any  $N\in\bbN$ and $\bm k, \bm h\in G^{\hat\Lambda_N^{(n_0,m_0)}}$,
\begin{align}
\begin{split}
&\varphi_{0}\lmk
\ket{\bm h}_{\hat\Lambda_N^{(n_0,m_0)}}\bra{\bm k}
\rmk\\
&=\sum_{\bm m\in G^{\hln {Nm_0}\setminus \hln N^{(n_0,m_0)}}}\varphi_{0}\lmk
\ket{\bm h,\bm m}_{\hln {Nm_0}}\bra{\bm k,\bm m}
\rmk\\
&=
\sum_{\bm m\in G^{\hln {Nm_0}\setminus \hln N^{(n_0,m_0)}}}
\frac{1}{|C_{\hat\Lambda_{Nm_0}}|}
\chi\lmk {\lmk \bm k,\bm m\rmk, \lmk \bm h,\bm m\rmk \in C_{\hat\Lambda_{Nm_0}}}\rmk
\chi\lmk \lmk \bm k,\bm m\rmk_{\partial{\hln {N m_0}}}=
\lmk \bm h,\bm m\rmk_{\partial{\hln {Nm_0}}}\rmk\\
&=
\sum_{\bm m\in G^{\hln {Nm_0}\setminus \hln N^{(n_0,m_0)}}}
\frac{1}{|C_{\hat\Lambda_{Nm_0}}|}
\chi\lmk {\bm k,  \bm h\in C_{\hat\Lambda_{N}^{(n_0,m_0)}}}\rmk
\chi\lmk \bm k_{\partial{\hln {N}^{(n_0,m_0)}}}=
 \bm h_{\partial{\hln {N}^{(n_0,m_0)}}}\rmk\\
&\chi\lmk {\lmk \bm h,\bm m\rmk \in C_{\hat\Lambda_{Nm_0}}}\rmk
\chi\lmk {\lmk \bm k,\bm m\rmk, \lmk \bm h,\bm m\rmk \in C_{\hat\Lambda_{Nm_0}}}\rmk
\chi\lmk \lmk \bm k,\bm m\rmk_{\partial{\hln {N m_0}}}=
\lmk \bm h,\bm m\rmk_{\partial{\hln {Nm_0}}}\rmk\\
&=
\sum_{\bm m\in G^{\hln {Nm_0}\setminus \hln N^{(n_0,m_0)}}}
\frac{1}{|C_{\hat\Lambda_{Nm_0}}|}
\chi\lmk {\bm k,  \bm h\in C_{\hat\Lambda_{N}^{(n_0,m_0)}}}\rmk
\chi\lmk \bm k_{\partial{\hln {N}^{(n_0,m_0)}}}=
 \bm h_{\partial{\hln {N}^{(n_0,m_0)}}}\rmk
 \chi\lmk {\lmk \bm h,\bm m\rmk \in C_{\hat\Lambda_{Nm_0}}}\rmk\\
 &=\frac{1}{|C_{\hat\Lambda_{Nm_0}}|}
\chi\lmk {\bm k,  \bm h\in C_{\hat\Lambda_{N}^{(n_0,m_0)}}}\rmk
\chi\lmk \bm k_{\partial{\hln {N}^{(n_0,m_0)}}}=
 \bm h_{\partial{\hln {N}^{(n_0,m_0)}}}\rmk
 \lv
 \left\{
 \bm m\in G^{\hln {Nm_0}\setminus \hln N^{(n_0,m_0)}}\mid
 \lmk \bm h,\bm m\rmk \in C_{\hat\Lambda_{Nm_0}}
 \right\}
 \rv\\
 &=\frac{1}{|C_{\hat\Lambda_{N}^{(n_0,m_0)}}|}
\chi\lmk {\bm k,  \bm h\in C_{\hat\Lambda_{N}^{(n_0,m_0)}}}\rmk
\chi\lmk \bm k_{\partial{\hln {N}^{(n_0,m_0)}}}=
 \bm h_{\partial{\hln {N}^{(n_0,m_0)}}}\rmk.
 \end{split}
\end{align}
Here we used (\ref{monb}) and (\ref{short}) for the last equality.

\end{proof}

\begin{lem}\label{tad}
Let $v\in \Lambda_{N-1}$ with $N\ge 3$.
Then for each $g\in G$, there exists 
a bijection $T_g^{(Nv)}: C_{\hln N}\to C_{\hln N}$ such that
\begin{align}
A_v^{(g)}\ket{\bm k}_{\hln N}
=\ket{T_g^{(Nv)}\bm k}_{\hln N},\quad \bm k\in C_{\hln N}
\end{align}
and
\begin{align}
\bm k_{\partial{\hln N} }=\lmk T_g^{(Nv)}\bm k\rmk_{\partial{\hln N}} .
\end{align}

\end{lem}
\begin{proof}
Let $v=(a, b)\in\Lambda_{N-1}$.
Then we have
\begin{align}
A_v^{(g)}\ket{\bm k}_{\hln N}
=\ket{T_g^{(Nv)}\bm k}_{\hln N},\quad \bm k\in C_{\hln N}
\end{align}
with
\begin{align}
\begin{split}
\lmk T_g^{(Nv)}\bm k\rmk_{e}
=
\left\{
\begin{gathered}
k_{(a-1,b), (a,b)}g^{-1},\quad e=((a-1,b), (a,b)),\\
k_{(a,b-1), (a,b)}g^{-1},\quad e=((a,b-1), (a,b)),\\
gk_{(a,b), (a,b+1)},\quad e=((a,b), (a,b+1)),\\
gk_{(a,b), (a+1,b)},\quad e=((a,b), (a+1,b)),\\
k_e,\quad \text{otherwise}
\end{gathered}
\right..
\end{split}
\end{align}
We then have
\begin{align}
\begin{split}
&\lmk T_g^{(Nv)}\bm k\rmk_{(a,b)-(a+1,b)}\lmk T_g^{(Nv)}\bm k\rmk_{(a+1,b)-(a+1,b+1)}=
gk_{(a,b), (a+1,b)}k_{(a+1,b)-(a+1,b+1)}\\
&=gk_{(a,b), (a,b+1)}k_{(a,b+1)-(a+1,b+1)}
=
\lmk T_g^{(Nv)}\bm k\rmk_{(a,b)-(a,b+1)}\lmk T_g^{(Nv)}\bm k\rmk_{(a,b+1)-(a+1,b+1)},\\
&\lmk T_g^{(Nv)}\bm k\rmk_{(a,b-1)-(a+1,b-1)}\lmk T_g^{(Nv)}\bm k\rmk_{(a+1,b-1)-(a+1,b)}=
k_{(a,b-1)-(a+1,b-1)}k_{(a+1,b-1)-(a+1,b)}\\
&=k_{(a,b-1), (a,b)}g^{-1}gk_{(a,b), (a+1,b)}
=
\lmk T_g^{(Nv)}\bm k\rmk_{(a,b-1)-(a,b)}\lmk T_g^{(Nv)}\bm k\rmk_{(a,b)-(a+1,b)},\\
&
\lmk T_g^{(Nv)}\bm k\rmk_{(a-1,b)-(a,b)}\lmk T_g^{(Nv)}\bm k\rmk_{(a,b)-(a,b+1)}=
k_{(a-1,b), (a,b)}g^{-1}gk_{(a,b), (a,b+1)}\\
&=k_{(a-1,b)-(a-1,b+1)}k_{(a-1,b+1)-(a,b+1)}
=
\lmk T_g^{(Nv)}\bm k\rmk_{(a-1,b)-(a-1,b+1)}\lmk T_g^{(Nv)}\bm k\rmk_{(a-1,b+1)-(a,b+1)},\\
&
\lmk T_g^{(Nv)}\bm k\rmk_{(a-1,b-1)-(a,b-1)}\lmk T_g^{(Nv)}\bm k\rmk_{(a,b-1)-(a,b)}
=k_{(a-1,b-1)-(a,b-1)}k_{(a,b-1), (a,b)}g^{-1},\\
&=k_{(a-1,b-1)-(a-1,b)}k_{(a-1,b), (a,b)}g^{-1}=
\lmk T_g^{(Nv)}\bm k\rmk_{(a-1,b-1)-(a-1,b)}\lmk T_g^{(Nv)}\bm k\rmk _{(a-1,b)-(a,b)},
\end{split}
\end{align}
and we have
$T_g^{(Nv)}\bm k\in C_{\hln N}$.
Because $\lmk A_v^{(g)}\rmk^{-1}=A_v^{(g^{-1})}$,
$T_g^{(Nv)}$ is a bijection.
Because $A_v^{(g)}$, $v\in\Lambda_{N-1}$ does not change the boundary edge, we have
\begin{align}
 \bm k_{\partial{\hln N}}=\lmk T_g^{(Nv)}\bm k\rmk_{\partial{\hln N} }.
\end{align}
\end{proof}
The state $\varphi_0$ is the  frustration-free ground state of the quantum double model.
\begin{lem}
We have
\begin{align}
\begin{split}
\varphi_0\lmk
A_v^{(g)}
\rmk=1
\end{split}
\end{align}
for all vertices $v\in\bbZ^2$ and $g\in G$.
In particular, 
for $A_v=\frac{1}{|G|}\sum_{g\in G}A_v^{(g)}$,
we have
\begin{align}
\varphi_0(A_v)=1
\end{align}
for all vertices $v\in\bbZ^2$ and
\begin{align}
\varphi_0(B_p)=1
\end{align}
for all squares $p$.
\end{lem}
\begin{proof}
By (\ref{vpdef}), we have
\begin{align}\label{bpp}
\begin{split}
\varphi_0\lmk\ket{\bm k}_{\hln N}\bra{\bm k}\rmk
=0,\quad
\bm k\in G^{\hln N}\setminus C_{\hln N},
\end{split}
\end{align}
hence $\varphi_0(B_p)=1$.

For any $v\in \bbZ^2$, choose
$3\le N\in\bbN$
so that $v\in \Lambda_{N-1}$. For any $g\in G$, we have
\begin{align}
\begin{split}
&\varphi_0\lmk
A_v^{(g)}
\rmk
=\sum_{\bm k\in G^{\hln N}}\varphi_0\lmk A_v^{(g)}\ket{\bm k}_{\hat\Lambda_N}\bra{\bm k}\rmk
=\sum_{\bm k\in C_{{\hln N}}}\varphi_0\lmk A_v^{(g)}\ket{\bm k}_{\hat\Lambda_N}\bra{\bm k}\rmk\\
&=\sum_{\bm k\in C_{{\hln N}}}\varphi_0\lmk 
\ket{T_g^{(Nv)}\bm k}_{\hat\Lambda_N}\bra{\bm k}\rmk
=\frac{1}{|C_{\hat\Lambda_N}|}\sum_{\bm k\in C_{{\hln N}}}
\chi\lmk {\bm k, T_g^{(Nv)}\bm k\in C_{\hat\Lambda_N}}\rmk
\chi\lmk \bm k_{\partial{\hln N}}= \lmk T_g^{(Nv)}\bm k\rmk_{\partial{\hln N}}\rmk=1.
\end{split}
\end{align}
We used (\ref{bpp}) for the second equality, and Lemma \ref{tad} for the third and fifth equality.
\end{proof}
Hence our $\varphi_0$ is the  frustration-free ground state of the quantum double model.
\begin{lem}
We have $\omega_0=\varphi_0$. Namely,
the restriction of $\omega_0$ onto 
$\caB_{\hln N}$ is given by the formula (\ref{vpdef}).
\end{lem}
%
%
%

\section{Admissible configurations on layers of squares}\label{acls}
In this section, we consider admissible configuration on layers of squares.
More precisely we consider the following shape.
We use the notation from section \ref{intro}, but we do not assume
$G$ to be abelian.
\begin{defn}\label{ds}
We consider $l$-layers $\mathfrak S$ of a sequence of squares $m=1,\ldots, l$ in $\bbZ^2$,
with $m$-th layer
\[
S_{(x_1^{(m)}, y+m-1)},  S_{(x_2^{(m)}, y+m-1)}, \ldots, S_{(x_{n_m}^{(m)}, y+m-1)},
\] 
with $x_k^{(m)}=x_1^{(m)}+k-1\in\bbZ$, $y\in\bbZ$.
We say this layer satisfies the condition S if
\begin{align}
\begin{split}
z^{(m)}:=\max\{x_{1}^{(m)}, x_1^{(m+1)}\}<\min\{x_{n_m}^{(m)}, x_{n_{m+1}}^{(m+1)} \}+1=:w^{(m)},
\end{split}
\end{align}
for all $m=1,\ldots,l-1$.
\end{defn}
Let us consider $l$-layers $\mathfrak S$ of a sequence of squares 
with notations in Definition \ref{ds}, satisfying the condition S.
We set 
\begin{align}
\begin{split}
&M^{(m)}:=w^{(m)}-z^{(m)}\in\bbN, \\
&\tilde z^{(m)}:=\min \{x_{1}^{(m)}, x_1^{(m+1)}\},\\
&\tilde w^{(m)}:=\max\{x_{n_m}^{(m)}, x_{n_{m+1}}^{(m+1)} \}+1,
\end{split}
\end{align}
for $\quad m=1,\ldots,l-1$.
We name the interior horizontal edges between $m$-th and $m+1$-th layer as
\begin{align}
\begin{split}
&\bm e_{k}^{({m})}:=(z^{(m)}+k-1, y+{m})-(z^{({m})}+k, y+{m}),\quad k=1,\ldots, M^{(m)}-1\\
\end{split}
\end{align}
for $1\le m\le l-1$ and
\begin{align}
\begin{split}
\bm f^{(m)}:=(w^{(m)}-1, y+{m})-(w^{(m)}, y+{m}).
\end{split}
\end{align}
for $1\le m\le l-1$.
We denote 
interior vertical edges as
\begin{align}
\begin{split}
&\bm {\tilde f}^{(m)}_{k}:=(x_{k}^{(m)}+1, y+m-1)-(x_{k}^{(m)}+1, y+m),\quad
k=1,\ldots, n_m-1
\end{split}
\end{align}
for $1\le m\le l$.
The edges on the boundary are
\begin{align}
\begin{split}
&\bm b_k^{(1)}:=(x_k^{(1)},y)-(x_{k}^{(1)}+1,y),\quad k=1,\ldots, n_1,\\
&\bm b_k^{(l+1)}:=(x_k^{(l)},y+l)-(x_{k}^{(l)}+1,y+l),\quad k=1,\ldots, n_l,
\end{split}
\end{align}
corresponding to the top and bottom lines
and
\begin{align}
\begin{split}
&\bm b_k^{(m L)}:=(\tilde{z}^{(m)}+k-1, y+{m})-(\tilde{z}^{(m)}+k, y+{m}),\quad
k=1,\ldots, z^{(m)}-\tilde z^{(m)}\\
&\bm b_k^{(m R)}:=(w^{(m)}+k-1, y+{m})-(w^{(m)}+k, y+{m}),\quad
k=1,\ldots, \tilde w^{(m)}-w^{(m)}
\end{split}
\end{align}
corresponding to the left and right segment of the horizontal line between $m$-th and $m+1$-th layer
$m=1,\ldots,l-1$,
and
\begin{align}
\begin{split}
&\tilde{\bm b}^{(m L)}:=(x_1^{(m)},y+m-1)-(x_1^{(m)},y+m),\\
&\tilde{\bm b}^{(m R)}:=(x_{n_m}^{(m)}+1,y+m-1)-(x_{n_m}^{(m)}+1,y+m),
\end{split}
\end{align}
$m=1,\ldots,l$
corresponding to the vertical line.

We set
\begin{align}\label{somen}
\begin{split}
&\caE^{(1)}\lmk {\mathfrak S}\rmk:=
\left\{\bm e_{k}^{({m})}
\right\},\quad
\caE^{(2)}\lmk {\mathfrak S}\rmk:=
\left\{\bm f^{(m)}
\right\},\quad
\caE^{(3)}\lmk {\mathfrak S}\rmk:=
\left\{
\tilde {\bm f}^{(m)}_{k}
\right\},\\
&\caE^{(4)}\lmk {\mathfrak S}\rmk:=
\left\{
\bm b_k^{(1)}, \;\bm b_k^{(l+1)},\;\bm b_k^{(m L)},\;
\bm b_k^{(m R)},\; \tilde{\bm b}^{(m L)},\;\tilde {\bm b}^{(m R)}
\right\},
\end{split}
\end{align}
Then we have
\begin{align}
\begin{split}
&\bbE\lmk\mathfrak S\rmk:=\cup_{p\in\mathfrak S}\{\text{edge of}\;p\}
=
 \caE^{(1)}\lmk {\mathfrak S}\rmk \dot{\cup} \caE^{(2)}\lmk {\mathfrak S}\rmk \dot{\cup} 
\caE^{(3)}\lmk {\mathfrak S}\rmk \dot{\cup} \caE^{(4)}\lmk {\mathfrak S}\rmk.
\end{split}
\end{align}
We also set
\begin{align}
&\mathbb V\lmk \mathfrak S \rmk
:=\cup_{e\in \mathbb E\lmk{\mathfrak S}\rmk}\{ v_{e,+1}, v_{e,-1}\}.
\end{align}


\begin{defn}\label{ringo}
Let  $\mathfrak S$ be $l$-layers of a sequence of squares given in Definition \ref{ds},
satisfying  the condition S.
We denote by $\mfp^{(r)}$ the self-avoiding path 
in $\caE^{(4)}\lmk {\mathfrak S}\rmk$ with origin
$(x_1^{(1)},y)$ target $(x_{n_{l}}^{({l})}+1, y+l)$
starting as 
$(x_1^{(1)}, y)-(x_1^{(1)}+1, y)-\cdots$ and following the boundary.
We denote by $\mfp^{(l)}$ the self-avoiding path in $\caE^{(4)}\lmk {\mathfrak S}\rmk$ with origin
$(x_1^{(1)},y)$ target $(x_{n_{l}}^{({l})}+1, y+l)$
 starting as 
$(x_1^{(1)}, y)-(x_1^{(1)}, y+1)-\cdots$ and following the boundary.
Because of condition S,
$\mfp^{(r)}$ and $\mfp^{(l)}$ intersects only at the origin 
$(x_1^{(1)},y)$ and the target $(x_{n_{l}}^{({l})}+1, y+l)$.
\end{defn}

\begin{lem}\label{basic}
Let  $\mathfrak S$ be $l$-layers of a sequence of squares given in Definition \ref{ds}.
Suppose that the $\mathfrak S$ satisfies the condition S.
Then
 for any $v_0,v_1\in \mathbb V\lmk \mathfrak S \rmk$ and paths
$\mfp,\tilde{\mfp}$ in $\bbE\lmk\mathfrak S\rmk$ with origin  $v_0$ and target $v_1$,
there exists a finite sequence of paths $\mfp_i $ in $\bbE(\mathfrak S)$ with origin  $v_0$ and target $v_1$
$i=1,\ldots, n$,
such that $\mfp_1=\mfp$, $\mfp_n=\tilde\mfp$ and
$\lmk \mfp_i\rmk^{-1}\mfp_{i+1}=\mfp_{S_i}$
or $\lmk \mfp_i\rmk^{-1}\mfp_{i+1}=\mfp_{S_i^{-1}}$
for some square $S_i$ in $\mathfrak S$
or $\lmk \mfp_i\rmk^{-1}\mfp_{i+1}=\mfq_i\mfq_i^{-1}$ for a path
$\mfq_i$  in $\bbE\lmk \mathfrak S\rmk$ , $i=1,\ldots,n-1$.
\end{lem}
\begin{rem}
Let $\tilde{\mathfrak S}$ be a set of squares and  $\bbE\lmk\tilde{\mathfrak S}\rmk$
the set of all edges of the squares in $\tilde{\mathfrak S}$.
For paths $\mfp,\tilde{\mfp}$ in $\bbE\lmk\tilde{\mathfrak S}\rmk$ with common
origin  and target,
we say $\mfp$ can be deformed into $\tilde\mfp$
in $\tilde{\mathfrak S}$ if there is a sequence of paths
as in Lemma \ref{basic}.
\end{rem}
\begin{proof}
We consider the following proposition for each $l\in\bbN$.
\begin{quote}
$P_l$ : For $l$-layers $\mathfrak S$ of a sequence of squares satisfying the condition S,
for any $v_0,v_1\in \mathbb V\lmk \mathfrak S \rmk$ and paths
$\mfp,\tilde{\mfp}$ in $\bbE\lmk\mathfrak S\rmk$ with origin  $v_0$ and target $v_1$,
there exists a finite sequence of paths $\mfp_i $ with origin  $v_0$ and target $v_1$
$i=1,\ldots, n$,
such that $\mfp_1=\mfp$, $\mfp_n=\tilde\mfp$ and
$\lmk \mfp_i\rmk^{-1}\mfp_{i+1}=\mfp_{S_i}$
or $\lmk \mfp_i\rmk^{-1}\mfp_{i+1}=\mfp_{S_i^{-1}}$
for some square $S_i$ in $\mathfrak S$
or $\lmk \mfp_i\rmk^{-1}\mfp_{i+1}=\mfq_i\mfq_i^{-1}$ for a path
$\mfq_i$  in $\bbE\lmk \mathfrak S\rmk$ , $i=1,\ldots,n-1$.
\end{quote}
$P_1$ is true because any path from $v_0=(a,b)$ to $v_1=(c,d)$
 in $1$-layer of squares can be deformed into
 a path
 \begin{center}
 $(a,b)\to (c,b)$ horizontally $(c,b)\to(c,d)$ vertically.
 \end{center} 
 Suppose $P_l$ is true.
 Let $\mathfrak S$ be $l+1$-layers of a sequence of squares satisfying the condition S.
Let $v_0,v_1\in \mathbb V\lmk \mathfrak S \rmk$, and 
let $\mfp,\tilde{\mfp}$ be paths in $\bbE\lmk\mathfrak S\rmk$
with origin  $v_0$ and target $v_1$.
We consider the case $v_0$ is in the first $l$-layers and $v_1$ is in the 
$l+1$-th layers.
The proof are the same for other cases.
The path $\mfp$ can be split into sequence of paths 
\[
\mfp_1, \hat\mfp_1, \mfp_2, \hat\mfp_2,\cdots \mfp_L, \hat\mfp_L, 
\]
where $\mfp_k$ is a path inside of the first $l$-layers
and $\hat\mfp_k$ is a path inside of the $l+1$-th layer.
The origin of $\mfp_1$ is $v_0$ and the target of $\hat{{\mfp}}_{ L}$ is $v_1$.
The target of $\mfp_k$ is origin of 
$\hat{\mfp}_k$, $k=1,\ldots,L$
and
the target of $\hat{\mfp}_k$ is the origin of $\mfp_{k+1}$,
$k=1,\ldots,L-1$.
For each $k$,  let $\mathfrak q_k$ $k=1,\ldots, L$ 
(resp $\hat{\mathfrak q}_k$, $k=1,\ldots, L-1$) be a horizontal path in $\mathfrak S$
in the line between $l$-th and $l+1$-th layer such that
\[
\text{origin of}\;  \mathfrak q_k=\text{target of}\; \mfp_k=\text{origin of}\; \hat\mfp_k\quad
\text{target of}\;\mathfrak q_k=(x_1^{(l)}, y+l-1)
\]
\[
\text{origin of}\;  \hat{\mathfrak q}_k=\text{target of}\; \hat{\mfp}_k=\text{origin of}\;\mfp_{k+1},\quad
\text{target of}\;\hat{\mathfrak q}_k=(x_1^{(l)}, y+l-1).
\]

The path
$\tilde{\mfp}$ can be split into sequence of paths 
\[
{\tilde{\mfp}}_1, \hat{\tilde{\mfp}}_1, {\tilde{\mfp}}_2, \hat{\tilde{\mfp}}_2,\cdots {\tilde{\mfp}}_{\tilde L}, \hat{\tilde{\mfp}}_{\tilde L}, 
\]
where ${\tilde{\mfp}}_k$ is a path inside of the first $l$-layers
and $\hat{\tilde{\mfp}}_k$ is a path inside of the $l+1$-th layer.
The target of $\tilde\mfp_k$ is origin of 
$\hat{\tilde\mfp}_k$, $k=1,\ldots,\tilde L$
and
the target of $\hat{\tilde\mfp}_k$ is the origin of ${\tilde\mfp}_{k+1}$,
$k=1,\ldots,\tilde L-1$.
The origin of ${\tilde{\mfp}}_1$ is $v_0$ and the
target of $\hat{\tilde{\mfp}}_{\tilde L}$ is $v_1$.
For each $k$, let $\tilde{\mathfrak q}_k$ $k=1,\ldots, \tilde L$  (resp. $\hat{\tilde{\mathfrak q}}_k$
$k=1,\ldots, \tilde L-1$) be 
a horizontal path in $\mathfrak S$
in the line between $l$-th and $l+1$-th layer such that
\[
\text{origin of}\;  {\tilde{\mathfrak q}_k}=\text{target of}\; {\tilde{\mfp}}_k
=\text{origin of}\;{ \hat{\tilde{\mfp}}}_k,\quad
\text{target of}\;{\tilde{\mathfrak q}_k}=(x_1^{(l)}, y+l-1)
\]
\[
\text{origin of}\;  \hat{\tilde{\mathfrak q}}_k=\text{target of}\; \hat{\tilde\mfp}_k
=\text{origin of}\;\tilde\mfp_{k+1},\quad
\text{target of}\;\hat{\tilde{\mathfrak q}}_k=(x_1^{(l)}, y+l-1).
\]

We fix a horizontal loop $\mathfrak l$ in $\mathfrak S$,
in the line between $l$-th and $l+1$-th layer,
with origin $(x_1^{(l)}, y+l-1)$ and target $(x_1^{(l)}, y+l-1)$.

Both of  $\mfp_1\mathfrak q_1$ and $\tilde{\mfp}_1\tilde{\mathfrak q}_1$ are paths inside of the first $l$-layers
with origin $v_0$ and target $(x_1^{(l)}, y+l-1)$.
Hence by $P_l$,  $\mfp_1\mathfrak q_1$ can be deformed into $\tilde{\mfp}_1\tilde{\mathfrak q}_1$
in the first $l$-layers. 
All of ${\mathfrak q}_j^{-1} \hat\mfp_j\hat{\mathfrak q}_j$, $j=1,\ldots, L-1$ and
$\tilde {\mathfrak q}_k^{-1} \hat{\tilde{\mfp}}_k\hat{\tilde{\mathfrak q}}_k$
$k=1,\ldots, \tilde L-1$
are paths inside of the $l+1$-th layer
 with origin $(x_1^{(l)}, y+l-1)$ and target $(x_1^{(l)}, y+l-1)$.
 Hence from $P_1$, 
 they can be deformed into $\mathfrak l$ in the $l+1$-th layer.
All of 
$\hat{\mathfrak q}_{j-1}^{-1} \mfp_j\mathfrak q_{j}$,  $j=2,\ldots, L$ and
$\hat{\tilde{\mathfrak q}}_{k-1}^{-1} \tilde\mfp_k\tilde{\mathfrak q}_{k}$ 
$k=2,\ldots, \tilde L$ are paths inside of 
 the first $l$-layers with origin $(x_1^{(l)}, y+l-1)$ and target $(x_1^{(l)}, y+l-1)$.
 Hence by $P_l$, they all
can be deformed into $\mathfrak l^{-1}$.
Both of  $\mfq_L^{-1}\hat\mfp_L$ and  $\tilde{\mfq}_{\tilde L}^{-1}\hat{\tilde\mfp}_{\tilde L}$ are paths in
$l+1$-th layer with origin $(x_1^{(l)}, y+l-1)$ target $v_1$. Hence by $P_1$, $\mfq_L^{-1}\hat\mfp_L$ 
can be deformed into $\tilde{\mfq}_{\tilde L}^{-1}\hat{\tilde\mfp}_{\tilde L}$ in the $l+1$-th layer.

 Hence $\mfp$ can be deformed into
 \begin{align}
 \begin{split}
 \tilde{\mfp}_1\tilde{\mathfrak q}_1, \mathfrak l, \mathfrak l^{-1},\cdots,
 \mathfrak l, \mathfrak l^{-1}, \tilde{\mfq}_{\tilde L}^{-1}\hat{\tilde\mfp}_{\tilde L}
 \end{split}
 \end{align}
 with $L-1$ number of $ \mathfrak l, \mathfrak l^{-1}$
 and 
$\tilde \mfp$ can be deformed into
 \begin{align}
 \begin{split}
 \tilde{\mfp}_1\tilde{\mathfrak q}_1, \mathfrak l, \mathfrak l^{-1},\cdots,
 \mathfrak l, \mathfrak l^{-1}, \tilde{\mfq}_{\tilde L}^{-1}\hat{\tilde\mfp}_{\tilde L},
 \end{split}
 \end{align}
 with $\tilde L-1$ number of $ \mathfrak l, \mathfrak l^{-1}$.
 Therefore, $\mfp$ can be deformed into $\tilde\mfp$
 in $\mathfrak S$.
\end{proof}
\begin{defn}\label{kaki}
Let $\mathfrak S$ be
$l$-layers  of a sequence of squares satisfying condition S.
Any two vertices in $v_0,v\in \bbV(\mathfrak S)$ can be connected via a path in
${\bbE(\mathfrak S)}$.
By Lemma \ref{basic}, 
for any paths $\mfp,\tilde \mfp$ in $\bbE(\mathfrak S)$ with origin
$v_0$ and target $v\in \bbV(\mathfrak S)$
can be deformed into each other in $\mathfrak S$.
Note (with notation in Lemma \ref{basic})
that for each i-th step of the deformation,
 we have $\Psi_{\mfp_i}(\bm h_{\mfp_i})=\Psi_{\mfp_{i+1}}(\bm h_{\mfp_{i+1}})$
for any $\bm h\in C_{\bbE(\mathfrak S)}$, by the admissibility condition.
Hence  for any  $\bm h\in C_{\bbE(\mathfrak S)}$ and $v_0\in \bbV(\mathfrak S)$,
any paths $\mfp,\tilde \mfp$ in $\bbE(\mathfrak S)$ with origin
 $v_0$ and target $v\in \bbV(\mathfrak S)$, we have
\begin{align}
\Psi_\mfp(\bm h_\mfp)=\Psi_{\tilde \mfp}(\bm h_{\tilde \mfp})
\end{align}
and
we may define
\begin{align}
\lmk \Psi^{v_0}(\bm h)\rmk_v:=\Psi_{ \mfp}(\bm h_{ \mfp}),\quad \bm h\in C_{\bbE(\mathfrak S)}, \; v\in {\mathbb V\lmk \mathfrak S \rmk},
\end{align}
independent of the choice of the path $\mfp$ in $\mathfrak S$ with origin $v_0$ and target $v$.

\end{defn}

We would like to count the number of admissible configurations
on $\bbE\lmk\mathfrak S\rmk$
under a given boundary condition.
We have the following Lemma.
\begin{lem}\label{alr}
Let  $\mathfrak S$ be $l$-layers of a sequence of squares.
Suppose that $\mathfrak S$ satisfies the condition S.
Let $\mfp^{(l)}$, $\mfp^{(r)}$ be the paths in Definition \ref{ringo} for this $\mathfrak S$.
For $\bm a=(a_v)_{v\in\caE^{(4)}\lmk {\mathfrak S}\rmk}\in G^{\caE^{(4)}\lmk {\mathfrak S}\rmk}$, the followings are
equivalent.
\begin{description}
\item[(i)]$\Psi_{\mfp^{(r)}}(\bm a)=\Psi_{\mfp^{(l)}}(\bm a)$.
\item[(ii)] There is a
\[
\bm g=\lmk g_e\rmk_{e\in {\caE^{(1)}\lmk {\mathfrak S}\rmk}}\in G^{\caE^{(1)}\lmk {\mathfrak S}\rmk}
\]
which allows unique $\bm b^{(2)}\in G^{\caE^{(2)}\lmk {\mathfrak S}\rmk}$, 
$\bm b^{(3)}\in G^{\caE^{(3)}\lmk {\mathfrak S}\rmk}$ such that
\[
(\bm a, \bm g, \bm b^{(2)}, \bm b^{(3)})\in C_{\bbE\lmk\mathfrak S\rmk}.
\]
\item[(iii)] For any 
\[
\bm g=\lmk g_e\rmk_{e\in {\caE^{(1)}\lmk {\mathfrak S}\rmk}}\in G^{\caE^{(1)}\lmk {\mathfrak S}\rmk}
\]
there is unique $\bm b^{(2)}\in G^{\caE^{(2)}\lmk {\mathfrak S}\rmk}$, 
$\bm b^{(3)}\in G^{\caE^{(3)}\lmk {\mathfrak S}\rmk}$ 
 such that
 \[
(\bm a, \bm g, \bm b^{(2)}, \bm b^{(3)})\in C_{\bbE\lmk\mathfrak S\rmk}.
\]
\end{description}
\end{lem}
\begin{proof}
We use the notation in Definition \ref{ds} to describe $\mathfrak S$.
(iii)$\Rightarrow$ (ii) is trivial.
(ii)$\Rightarrow$ (i) holds from the admissibility condition
of $(\bm a, \bm g, \bm b^{(2)}, \bm b^{(3)})\in C_{\bbE\lmk\mathfrak S\rmk}$,
because from Lemma \ref{basic}
we can deform the path $\mfp^{(l)}$ into $\mfp^{(r)}$
via paths inside of $\bbE\lmk\mathfrak S\rmk$ connecting $(x_1^{(1)},y)$ to {$(x_{n_{l_\Gamma}}^{({l_\Gamma})}+1, y+{l_\Gamma})$}.
\\
Now we prove (i)$\Rightarrow$ (iii).
Fix any $\bm g=\lmk g_e\rmk_{e\in {\caE^{(1)}\lmk {\mathfrak S}\rmk}}\in G^{{\caE^{(1)}\lmk {\mathfrak S}\rmk}}$.
We show that the label of rest of the edges $\caE^{(2)}\lmk {\mathfrak S}\rmk\cup\caE^{(3)}\lmk {\mathfrak S}\rmk$
are determined uniquely by the admissibility condition.
For each $m=1,\ldots, l$, we denote by $\mathfrak S^{(m)}$
the $m$-th layer of the squares of $\mathfrak S$.

(1)
We start from $\mathfrak S^{(1)}$.
The  condition
$\bm g=\lmk g_e\rmk_{e\in {\caE^{(1)}\lmk {\mathfrak S}\rmk}}\in G^{{\caE^{(1)}\lmk {\mathfrak S}\rmk}}$ and 
$\bm a=(a_v)_{v\in\caE^{(4)}\lmk {\mathfrak S}\rmk}\in G^{\caE^{(4)}\lmk {\mathfrak S}\rmk}$ set all
the labels of edges in $\mathfrak S^{(1)}$ : 
The edges whose labels are not determined yet are
\begin{align}
\begin{split}
\bm f^{(1)}:=(w^{(1)}-1, y+{1})-(w^{(1)}, y+{1}).
\end{split}
\end{align}
and
\begin{align}
\begin{split}
&\bm {\tilde f}^{(1)}_{k}:=(x_{k}^{(1)}+1, y)-(x_{k}^{(1)}+1, y+1),\quad
k=1,\ldots, n_1-1.
\end{split}
\end{align}
If $n_1\ge 2$ and $x_1^{(1)}+1\le w^{(1)}-1$, $ \bm {\tilde f}^{(1)}_{1}$ belongs to a square 
$S_{(x_1^{(1)}, y)}$.
All other edges in $S_{(x_1^{(1)}, y)}$
are in $\caE^{(4)}\lmk {\mathfrak S}\rmk\cup \caE^{(1)}\lmk {\mathfrak S}\rmk$.
The labels of these edges are determined from $\bm g$ and $\bm a$.
Therefore, by the admissibility condition,
the label of $ \bm {\tilde f}^{(1)}_{1}$ is determined automatically.
Next if $x_{2}^{(1)}+1\le w^{(1)}-1$ then
 $ \bm {\tilde f}^{(1)}_{1}$ and $ \bm {\tilde f}^{(2)}_{1}$
belong to the same square $S_{(x_2^{(1)}, y)}$.
Two other edges from this square $S_{(x_2^{(1)}, y)}$ 
belongs to $\caE^{(4)}\lmk {\mathfrak S}\rmk\cup \caE^{(1)}\lmk {\mathfrak S}\rmk$.
Hence three edges in this square is already determined.
By the admissibility condition
the label of the last one $ \bm {\tilde f}^{(2)}_{1}$ is uniquely determined.
We can continue this $w^{(1)}-x_1^{(1)}-1$ times
and determine labels of $ \bm {\tilde f}^{(1)}_{1},\cdot,\bm {\tilde f}^{(1)}_{w^{(1)}-x_1^{(1)}-1}$
uniquely.
If $w^{(1)}\le x_{n_1}^{(1)}$,  then
$x_{n_2}^{(1)}\le x_{n_1}^{(1)}$ and 
$\bm {\tilde f}^{(1)}_{n_1-1}$ belongs to a square $S_{(x_{n_1}^{(1)}, y)}$.
All other edges of $S_{(x_{n_1}^{(1)}, y)}$
are in $\caE^{(4)}\lmk {\mathfrak S}\rmk$.
The labels of these edges are determined from $\bm a$.
Therefore, by the admissibility condition,
the label of $ \bm {\tilde f}^{(1)}_{n_1}$ is determined automatically.
By the same procedure as before, 
all the labels of $\bm {\tilde f}^{(1)}_{k}$, $k=n_1-1, n_1-2,\ldots, w^{(1)}-x_1^{(1)}$
can be decided uniquely
so that the admissibility condition on
\begin{align}
\begin{split}
S_{(x_1^{(1)}, y)}, S_{(x_2^{(1)}, y)},\ldots, S_{(w^{(1)}-2, y)}, S_{(w^{(1)}, y)},\ldots, S_{(x_{n_1}^{(1)}, y)}
\end{split}
\end{align}
(i.e., all the squares in $\mathfrak S^{(1)}$ but $S_{(w^{(1)}-1, y)}$
are satisfied.
Hence all the edges in $\mathfrak S^{(1)}$ but $\bm f^{(1)}$ is labeled.
But $\bm f^{(1)}$ belongs to the square
$S_{(w^{(1)}, y)}$.Three other edges of $S_{(w^{(1)}, y)}$
are already labeled.
Hence from the admissibility condition of $S_{(w^{(1)}, y)}$,
this label $\bm f^{(1)}$  is also determined uniquely.

(2) This procedure continues.
By the $m$-th step, all the labels of edges in $\mathfrak S^{(1)},\ldots, \mathfrak S^{(m-1)}$
are determined uniquely, so that the admissibility conditions of all the squares in $\mathfrak S^{(1)},\ldots, \mathfrak S^{(m-1)}$
are satisfied.
In particular, all the edges in $\mathfrak S^{(m)}$
 except for $\bm f^{(m)}$, $\bm {\tilde f}^{(m)}_{k}:=(x_{k}^{(m)}+1, y+m-1)-(x_{k}^{(m)}+1, y+m)$,
$k=1,\ldots, n_m-1$ are already labeled.
By the same procedure as in (1), we can decide the label
of $\bm {\tilde f}^{(m)}_{k}$, $k=1,\ldots, n_m-1$ uniquely
so that the admissibility condition of
\begin{align}
\begin{split}
S_{(x_1^{(m)}, y+m-1)}, S_{(x_2^{(m)}, y+m-1)},
\ldots, S_{(w^{(m)}-2, y+m-1)}, S_{(w^{(m)}, y+m-1)},\ldots, S_{(x_{n_m}^{(m)}, y+m-1)}
\end{split}
\end{align}
(i.e., all the squares in $\mathfrak S^{(m)}$ but $S_{(w^{(m)}-1, y+m-1)}$)
are satisfied.
The label of $\bm {f}^{(m)}$ is the determined by
the admissibility condition on $S_{(w^{(m)}-1, y+m-1)}$ uniquely,
from the labels of edges already determined.

(3) Hence we label all the edges in $\mathfrak S^{(1)},\ldots, \mathfrak S^{(l-1)}$,
uniquely so that the admissibility conditions hold for all squares in $\mathfrak S^{(1)},\ldots, \mathfrak S^{(l-1)}$.
In $\mathfrak S^{(l)}$, the only edges without labels are
$\bm {\tilde f}^{(l)}_{k}$, $k=1,\ldots, n_l-1$.
We proceed in the order
\begin{align}
\begin{split}
S_{(x_1^{(l)}, y+l-1)}, S_{(x_2^{(l)}, y+l-1)},\ldots, S_{(x_{n_l}^{(l)}-1, y+l-1)}
\end{split}
\end{align}
 to define the
labels $\bm {\tilde f}^{(l)}_{k}$, $k=1,\ldots, n_l-1$
uniquely via the admissibility condition of these squares.
Hence we obtain the unique  label $\bm b^{(2)}\in \caE^{(2)}\lmk {\mathfrak S}\rmk$, 
$\bm b^{(3)}\in \caE^{(3)}\lmk {\mathfrak S}\rmk$
which satisfies the admissibility conditions of all the squares in $\mathfrak S$ but $S_{(x_{n_l}^{(l)}, y+l-1)}$.
The necessary and sufficient condition
for the admissibility condition of $S_{(x_{n_l}^{(l)}, y+l-1)}$
to hold for the obtained label $\bm k:=(\bm a, \bm g, \bm b^{(2)}, \bm b^{(3)})\in G^{\bbE(\mathfrak S)}$
is
\begin{align}
\begin{split}
&a_{(x_{n_{l}}^{({l})}+1,y+l-1)-(x_{n_{l}}^{({l})}+1,y+l )}\lmk \Psi_{\mfp^{(r)}}(\bm a)\rmk^{-1}
\Psi_{\mfp^{(l)}}(\bm a)
\lmk a_{(x_{n_{l}}^{({l})},y+l )-(x_{n_{l}}^{({l})}+1,y+l )}\rmk^{-1}\\
&=
\lmk k_{(x_{n_l}, y+l-1)-(x_{n_l}+1, y+l-1)}\rmk^{-1}k_{(x_{n_l}, y+l-1)-(x_{n_l}, y+l)
}\\
&=
a_{(x_{n_{l}}^{({l})}+1,y+l-1)-(x_{n_{l}}^{({l})}+1,y+l )}
\lmk a_{(x_{n_{l}}^{({l})},y+l )-(x_{n_{l}}^{({l})}+1,y+l )}\rmk^{-1}.
\end{split}
\end{align}
For the first equality, we used the fact that
 $\bm k:=(\bm a, \bm g, \bm b^{(2)}, \bm b^{(3)})\in G^{\bbE(\mathfrak S)}$
 satisfies the admissibility conditions of all the squares but $S_{(x_{n_l}^{(l)}, y+l-1)}$.
This condition is equivalent to (i).
Hence all the admissibility conditions are satisfied.

\end{proof}

\section{Admissible configurations in areas surrounded by a certain type of loops}\label{kuma}
In this section, we count the number of admissible configurations
in the area surrounded by a certain type of loops.
We use the notation from section \ref{intro}, \ref{coffe}, \ref{acls}
but we do not assume
$G$ to be abelian.

First we specify the kind of loops we consider.
By a path with origin $(x,y)\in\bbZ^2$ going up-right direction, we mean
a path of edges $\mathfrak p$ proceeding as follows
\begin{align}\label{rpt}
\begin{split}
&(x, y)-(x, y+1)-\cdots-(x, y+m_1),\quad\text{move vertically up},\\
&(x, y+m_1)-(x+1, y+m_1)-\cdots- (x+l_1, y+m_1)\quad\text{move horizontally right},\\
&(x+l_1, y+m_1)-(x+l_1, y+m_1+1)-\cdots-(x+l_1, y+m_1+m_2),
\quad\text{move vertically up},\\
&\cdots\\
&(x+l_1+\cdots l_k, y+m_1+\cdots m_k)
-(x+l_1+\cdots l_k, y+m_1+\cdots m_k+1)-\\
&-\cdots-(x+l_1+\cdots l_k, y+m_1+\cdots m_k+m_{k+1})
\quad\text{move vertically up},\\
&(x+l_1+\cdots l_k, y+m_1+\cdots m_k+m_{k+1})-
(x+l_1+\cdots l_k+1, y+m_1+\cdots m_k+m_{k+1})
-\\&-\cdots-
(x+l_1+\cdots l_k+l_{k+1}, y+m_1+\cdots m_k+m_{k+1})
\quad\text{move horizontally right},\\
&\cdots
\end{split}
\end{align}
with 
$m_1\in\bbN\cup\{0\}$, $l_1\in \bbN\cup\{0\}$
$m_k\in \bbN$, $l_k\in \bbN$ $k\ge 2$.
When $m_1=0$, we understand that the path first moves horizontally right
before going vertically up.
If $l_1=0$, the path never goes up horizontally and parallel to $y$-axis.
Note that all the edges have $+1$ directions along the path.
We denote by $\mathfrak P_{u,r}(x,y)$ the set of all
infinitely long paths with origin $(x,y)\in\bbZ^2$ going up-right direction.
We define the set of all
paths with origin $(x,y)\in\bbZ^2$ going up-left direction  $\mathfrak P_{u,l}(x,y)$, 
down-right direction $\mathfrak P_{d,r}(x,y)$, down-left direction $\mathfrak P_{d,l}(x,y)$ analogously.
Set $\mathfrak P(x,y):=\mathfrak P_{u,r}(x,y)\cup \mathfrak P_{u,l}(x,y)\cup \mathfrak P_{d,r}(x,y)\cup \mathfrak P_{d,l}(x,y)$.
For each $\mfp\in \mathfrak P(x,y)$, we denote the
corresponding parameters $m_k$, $l_k$s above
by $m_k^{\mfp}$, $l_k^{\mfp}$.

Let $\mathfrak p\in \mathfrak P_{u,r}(x,y)$ with description (\ref{rpt}).
We attach it a sequence of squares from below and above.
The sequence from below is defined as follows.
For each portion of $\mfp$ in (\ref{rpt})
\begin{align}
\begin{split}
&(x+l_1+\cdots l_k, y+m_1+\cdots m_k)
-(x+l_1+\cdots l_k, y+m_1+\cdots m_{k}+1)-\\
&-\cdots-(x+l_1+\cdots l_k, y+m_1+\cdots m_k+m_{k+1})
\quad\text{move vertically},\\
\end{split}
\end{align}
we attach squares
\begin{align}
\begin{split}
&S_{(x+l_1+\cdots l_k, y+m_1+\cdots m_k)}-
S_{(x+l_1+\cdots l_k, y+m_1+\cdots m_{k}+1)}-\\
&-\cdots-S_{(x+l_1+\cdots l_k, y+m_1+\cdots m_k+m_{k+1}-1)}.
\end{split}
\end{align}
For each  portion of $\mfp$ in (\ref{rpt})
\begin{align}
\begin{split}
&(x+l_1+\cdots l_k, y+m_1+\cdots m_k+m_{k+1})-
(x+l_1+\cdots l_k+1, y+m_1+\cdots m_k+m_{k+1})
-\\&-\cdots-
(x+l_1+\cdots l_k+l_{k+1}, y+m_1+\cdots m_k+m_{k+1})
\quad\text{move horizontally},
\end{split}
\end{align}
we attach squares
\begin{align}
\begin{split}
&S_{(x+l_1+\cdots l_k+1, y+m_1+\cdots m_k+m_{k+1}-1)}
-S_{(x+l_1+\cdots l_k+2, y+m_1+\cdots m_k+m_{k+1}-1)}-\\
&-\cdots-
S_{(x+l_1+\cdots l_k+l_{k+1}-1, y+m_1+\cdots m_k+m_{k+1}-1)}-
S_{(x+l_1+\cdots l_k+l_{k+1}, y+m_1+\cdots m_k+m_{k+1}-1)}.
\end{split}
\end{align}
We denote by $S_i^{b, \mfp}$, $i\in\bbN$ the 
sequence obtained in this way from $\mfp\in \mathfrak P_{(u,r)}(x,y)$.
 
For the sequence from above, we set as follows.
For each portion of $\mfp$ in (\ref{rpt})
\begin{align}
\begin{split}
&(x+l_1+\cdots l_k, y+m_1+\cdots m_k)
-(x+l_1+\cdots l_k, y+m_1+\cdots m_{k}+1)-\\
&-\cdots-(x+l_1+\cdots l_k, y+m_1+\cdots m_k+m_{k+1})
\quad\text{move vertically},\\
\end{split}
\end{align}
we attach squares
\begin{align}
\begin{split}
&S_{(x+l_1+\cdots l_k-1, y+m_1+\cdots m_k)}-
S_{(x+l_1+\cdots l_k-1, y+m_1+\cdots m_{k}+1)}-\\
&-\cdots-S_{(x+l_1+\cdots l_k-1, y+m_1+\cdots m_k+m_{k+1}-1)}.
\end{split}
\end{align}
For each  portion of $\mfp$ in (\ref{rpt})
\begin{align}
\begin{split}
&(x+l_1+\cdots l_k, y+m_1+\cdots m_k+m_{k+1})-
(x+l_1+\cdots l_k+1, y+m_1+\cdots m_k+m_{k+1})
-\\&-\cdots-
(x+l_1+\cdots l_k+l_{k+1}, y+m_1+\cdots m_k+m_{k+1})
\quad\text{move horizontally},
\end{split}
\end{align}
we attach squares
\begin{align}
\begin{split}
&S_{(x+l_1+\cdots l_k-1, y+m_1+\cdots m_k+m_{k+1})}
-S_{(x+l_1+\cdots l_k, y+m_1+\cdots m_k+m_{k+1})}-\\
&-\cdots-
S_{(x+l_1+\cdots l_k+l_{k+1}-2, y+m_1+\cdots m_k+m_{k+1})}.
\end{split}
\end{align}
We denote by $S_i^{a, \mfp}$, $i\in\bbN$ the 
sequence obtained in this way from $\mfp\in \mathfrak P_{(u,r)}(x,y)$.

Analogously, for each $\mfp\in \mathfrak P_{u,l}(x,y)\cup \mathfrak P_{d,r}(x,y)\cup \mathfrak P_{d,l}(x,y)$,
we obtain sequences of squares
from below and above $S_i^{b, \mfp}$, $i\in\bbN$ $S_i^{a, \mfp}$, $i\in\bbN$.

Now we fix $v_0=(x,y)\in\bbZ^2$ and take
$\mfp_1,\mfp_2\in \mathfrak P(x,y)$.
Each of $\mfp_1,\mfp_2$ are associated with a sequence of squares
$\{S_i^{b, \mfp}\}_i$, $\{S_i^{a, \mfp}\}_i$ as above.
Let $n_0,m_0\in\bbN$ and $N\in \bbN$ be large enough so that $v\in \Lambda_{N-3}^{(n_0,m_0)}$.
Then each of $\mfp_1,\mfp_2$ 
intersects with $\partial\hat\Lambda_N^{(n_0,m_0)}$.
Let $w_{\mfp_1}^{(N)}$ (resp. $w_{\mfp_2}^{(N)}$)
 be the first vertex that $\mfp_1$ (resp.$\mfp_2$) 
intersects with $\partial\hat{\Lambda}_N^{(n_0,m_0)}$ when we proceed from $v_0=(x,y)$.
We denote by $\mfp_1^{(N)}$ (resp. $\mfp_2^{(N)}$) the portion of $\mfp_1$ (resp. $\mfp_2$)
from $v_0$ to $w_{\mfp_1}^{(N)}$ (resp $w_{\mfp_2}^{(N)}$).
We denote by $\mathfrak l_{\mfp_1,\mfp_2,+1}^{(N)}$
(resp. $\mathfrak l_{\mfp_1,\mfp_2,-1}^{(N)}$)
the path of edges on $\partial{{\hln N}^{(n_0,m_0)}}$
with origin $w_{\mfp_1}^{(N)}$
and the terminal $w_{\mfp_2}$, given by
proceeding along $\partial{{\hln N}^{(n_0,m_0)}}$ from $w_{\mfp_1}^{(N)}$ to $w_{\mfp_2}^{(N)}$
counter-clockwise (resp. clock-wise).
We obtain a closed loop of edges $\mathfrak c_{\mfp_1,\mfp_2,\sigma}^{(N)}$, $\sigma=\pm 1$
as follows:
Start from $v_0$. Proceed along $\mfp_1$ from $v_0$ to $w_{\mfp_1}^{(N)}$.
Proceed along $\mathfrak l_{\mfp_1,\mfp_2,\sigma}^{(N)}$
from $w_{\mfp_1}^{(N)}$ to $w_{\mfp_2}^{(N)}$.
Proceed  along  ${{\mfp_2}^{(N)}}^{-1}$,
from $w_{\mfp_2}^{(N)}$ to $v_0$.

\begin{defn}\label{ws}
Let $v_0=(x,y)\in \bbZ^2$ and 
$N_0, n_0,m_0\in\bbN$.
We say two paths $\mfp_1, \mfp_2\in \mathfrak P(x,y)$ of edges are well-separated
for with respect to $N_0, n_0,m_0$
if for any $N_0\le N\in\bbN$, 
$(x,y)\in {\Lambda_N^{(n_0,m_0)}}$ holds and
the loop $\mathfrak c_{\mfp_1,\mfp_2,\sigma}^{(N)}=
\mfp_1^{(N)}\mathfrak l_{\mfp_1,\mfp_2,\sigma}^{(N)}\lmk \mfp_2^{(N)}\rmk^{-1}$
is a simple closed loop and the area inside of
it consists of a set of squares satisfying the condition S.
We denote by 
$\mathfrak S_{\mfp_1,\mfp_2,\sigma}^{(N)}$
the set of all squares inside of this closed loop $\mathfrak c_{\mfp_1,\mfp_2,\sigma}^{(N)}=\mfp_1^{(N)}\mathfrak l_{\mfp_1,\mfp_2,\sigma}^{(N)}\lmk \mfp_2^{(N)}\rmk^{-1}$.
\end{defn}

\begin{set}\label{kame}
Let $v_0=(x,y)\in \bbZ^2$ and 
$N_0, n_0,m_0\in\bbN$.
We consider three cases:
\begin{description}
\item[(1)] Well-separated $\mfp_1,\mfp_2\in  \mathfrak P_{(u,r)}(x,y)$ with respect to $N_0, n_0,m_0$.
The path $\mfp_2$ is above $\mfp_1$, i.e.,
for all vertex $(s,t_1)$ on $\mfp_1$,
any vertex on $\mfp_2$ of the form $(s,t_2)$ satisfies $t_1<t_2$.

Note that because $\mfp_1$,$\mfp_2$ are well-separated, $m_{\mfp_1}=0$ and $m_{\mfp_2}\neq0$.
Because $(x, y-1)\notin \bbV \lmk\mathfrak S_{\mfp_1,\mfp_2,+1}^{(N)}\rmk$,
we have 
$(x,y)\in \partial\bbV \lmk\mathfrak S_{\mfp_1,\mfp_2,+1}^{(N)}\rmk$.
The set of squares
\begin{align}
\begin{split}
\lmk \{S_i^{b,\mfp_1}\}_{i}\cup \{S_{(x,y-1)}\}\cup \{S_{(x-1,y-1)}\}\cup \{S_i^{a, \mfp_2}\}_{i}\rmk\cap \caP_{{\Lambda_N^{(n_0,m_0)}}}
\end{split}
\end{align}
forms a sequence of squares $\mathfrak T_{\mfp_1,\mfp_2}^{(N)}:=(S_i^{\mfp_1,\mfp_2, (N)})$ along the path $\lmk\mfp_1^{(N)}\rmk^{-1}\mfp_2^{(N)}$.

\item[(2)]Well-separated $\mfp_1\in  \mathfrak P_{(u,r)}(x,y)$, $\mfp_2\in  \mathfrak P_{(u,l)}(x,y)$, with respect to $N_0, n_0,m_0$.
Note that  $m_{\mfp_1}=0$ or   $m_{\mfp_2}=0$.
Because $(x, y-1)\notin \bbV \lmk\mathfrak S_{\mfp_1,\mfp_2,+1}^{(N)}\rmk$,
we have 
$(x,y)\in \partial\bbV \lmk\mathfrak S_{\mfp_1,\mfp_2,+1}^{(N)}\rmk$.
The set of squares
\begin{align}
\begin{split}
\lmk \{S_i^{b,\mfp_1}\}_{i}\cup \{S_{(x,y-1)}\}\cup \{S_{(x-1,y-1)}\}\cup \{S_i^{b, \mfp_2}\}_{i}\rmk\cap \caP_{{\Lambda_N^{(n_0,m_0)}}},
\end{split}
\end{align}
forms a sequence of squares $\mathfrak T_{\mfp_1,\mfp_2}^{(N)}:=(S_i^{\mfp_1,\mfp_2, (N)})$ along the path $\lmk\mfp_1^{(N)}\rmk^{-1}\mfp_2^{(N)}$.
\item[(3)]Well-separated $\mfp_1\in  \mathfrak P_{(u,r)}(x,y)$, $\mfp_2\in  \mathfrak P_{(d,l)}(x,y)$
with respect to $N_0, n_0,m_0$.
We assume $m_{\mfp_2}=0$.
Because of $m_{\mfp_2}=0$, 
$(x, y-1)\notin \bbV \lmk\mathfrak S_{\mfp_1,\mfp_2,+1}^{(N)}\rmk$, and
we have 
$(x,y)\in \partial\bbV \lmk\mathfrak S_{\mfp_1,\mfp_2,+1}^{(N)}\rmk$.
The set of squares
\begin{align}
\begin{split}
\lmk \{S_i^{b,\mfp_1}\}_{i}\cup \{S_{(x,y-1)}\}\cup \{S_{(x-1,y-1)}\}\cup \{S_i^{b, \mfp_2}\}_{i}\rmk\cap \caP_{{\Lambda_N^{(n_0,m_0)}}}
\end{split}
\end{align}
forms a sequence of squares $\mathfrak T_{\mfp_1,\mfp_2}^{(N)}:=(S_i^{\mfp_1,\mfp_2, (N)})_{i=1}^{M_{\mfp_1,\mfp_2, (N)}}$ along the path $\lmk\mfp_1^{(N)}\rmk^{-1}\mfp_2^{(N)}$.
\end{description}
\end{set}
We  introduce
\begin{align}
\begin{split}
&\mathfrak {BE}_{\mfp_1,\mfp_2}^{(N)}:=
\left\{
e\in \hat{\Lambda}_N^{(n_0,m_0)}\mid
\lmk v_{e,-1}\in \mathbb V\lmk \mathfrak S_{\mfp_1,\mfp_2,+1}^{(N)} \rmk, v_{e,+1}\notin
\mathbb V\lmk \mathfrak S_{\mfp_1,\mfp_2,+1}^{(N)} \rmk 
\rmk
\text{or}
\lmk  v_{e,+1}\in \mathbb V\lmk \mathfrak S_{\mfp_1,\mfp_2,+1}^{(N)} \rmk,
 v_{e.-1}\notin
\mathbb V\lmk \mathfrak S_{\mfp_1,\mfp_2,+1}^{(N)} \rmk \rmk
\right\},\\
\end{split}
\end{align}
and 
\begin{align}
\begin{split}
&\mathfrak {OE}_{\mfp_1,\mfp_2}^{(N)}:=
\hat{\Lambda}_N^{(n_0,m_0)}\setminus\lmk
 \mathfrak {BE}_{\mfp_1,\mfp_2}^{(N)}\cup \mathbb E\lmk \mathfrak S_{\mfp_1,\mfp_2,+1}^{(N)} \rmk\rmk
\end{split}
\end{align}
If all the vertices of a square $S$ in 
$\Lambda_N^{(n_0, m_0)}$  belong to 
$\lmk \bbV\lmk \mathfrak S_{\mfp_1,\mfp_2,+1}^{(N)} \rmk\rmk^c$,
then 
all the edges of $S$ belong to $\mathfrak {OE}_{\mfp_1,\mfp_2}^{(N)}$.
\begin{rem}\label{obs}
Under Setting \ref{kame},
 the following can be checked from the construction.
\begin{description}
\item[(A)]
For each $i=1,\ldots, M_{\mfp_1,\mfp_2, (N)}-1$,
$S_i^{\mfp_1,\mfp_2, (N)}$ and $S_{i+1}^{\mfp_1,\mfp_2, (N)}$ shares exactly one edge $ e_{i}^{\mfp_1,\mfp_2, (N)}$. This edge $ e_{i}^{\mfp_1,\mfp_2, (N)}$ belongs to $\mathfrak {BE}_{\mfp_1,\mfp_2}^{(N)}$ and 
$e_{i}^{\mfp_1,\mfp_2, (N)}\notin \partial{{\hln N}^{(n_0,m_0)}}$.
\item[(B)]
There exists exactly one edge $e_{ M_{\mfp_1,\mfp_2, (N)}}^{\mfp_1,\mfp_2, (N)}$ in 
$S_{M_{\mfp_1,\mfp_2, (N)}}^{\mfp_1,\mfp_2, (N)}\cap\partial {{\hln N}^{(n_0,m_0)}}$
 satisfying the condition $e_{ M_{\mfp_1,\mfp_2, (N)}}^{\mfp_1,\mfp_2, (N)}\in \mathfrak {BE}_{\mfp_1,\mfp_2}^{(N)}$.
 There exists exactly one edge $e_0^{\mfp_1,\mfp_2, (N)}$ in 
$S_{1}^{\mfp_1,\mfp_2, (N)}\cap\partial {{\hln N}^{(n_0,m_0)}}$
 satisfying the condition $e_{ 0}^{\mfp_1,\mfp_2, (N)}\in \mathfrak {BE}_{\mfp_1,\mfp_2}^{(N)}$.
\item[(C)]
$S_i^{\mfp_1,\mfp_2, (N)}\cap S_{j}^{\mfp_1,\mfp_2, (N)}=\emptyset$
for $i+1<j$.
In particular, 
combining with (A), (B), $e_{i}^{\mfp_1,\mfp_2, (N)}\neq e_{i-1}^{\mfp_1,\mfp_2, (N)}$
for $1\le i\le M_{\mfp_1,\mfp_2, (N)}$.
\item[(D)]
Any square $S_i^{\mfp_1,\mfp_2, (N)}$ in $\mathfrak T_{\mfp_1,\mfp_2}^{(N)}$
has exactly two distinct edges $e_{i-1}^{\mfp_1,\mfp_2, (N)}, e_{i}^{\mfp_1,\mfp_2, (N)}$
in $\mathfrak {BE}_{\mfp_1,\mfp_2}^{(N)}$
and two other edges of $S_i^{\mfp_1,\mfp_2, (N)}$ are in 
$\mathbb E\lmk \mathfrak S_{\mfp_1,\mfp_2,+1}^{(N)} \rmk
\cup \mathfrak {OE}_{\mfp_1,\mfp_2}^{(N)} $.
Furthermore, we have $\mathfrak {BE}_{\mfp_1,\mfp_2}^{(N)}=\{ e_{i}^{\mfp_1,\mfp_2, (N)}\}_{i=0}^{M_{\mfp_1,\mfp_2, (N)}}$.
\item[(E)] For any $v\in \partial \mathbb V\lmk \mathfrak S_{\mfp_1,\mfp_2,+1}^{(N)} \rmk\setminus\partial{\Lambda_N^{(n_0,m_0)}}$,
there exists a path $\mfp_v$ in ${{\hln N}^{(n_0,m_0)}}\setminus \mathbb E\lmk \mathfrak S_{\mfp_1,\mfp_2,+1}^{(N)} \rmk$
with origin $v$ and target $v_0=(x,y)$.
\item[(F)]For $v\in\partial{\Lambda_N^{(n_0,m_0)}}$, there is a path $\mfp_v^1$ along $\partial{{\hln N}^{(n_0,m_0)}}$
with origin $v$ and target $(Nn_0,-Nm_0)$.
There is a path $\mfp_0$ in ${{\hln N}^{(n_0,m_0)}}\setminus \mathbb E\lmk \mathfrak S_{\mfp_1,\mfp_2,+1}^{(N)} \rmk$
with origin $(Nn_0,-Nm_0)$ target $v_0$.
\item[(G)]
Any square which is not in $\mathbb E\lmk \mathfrak S_{\mfp_1,\mfp_2,+1}^{(N)} \rmk$
nor $(S_i^{\mfp_1,\mfp_2, (N)})_{i=1}^{M_{\mfp_1,\mfp_2, (N)}}$
does not have vertices from $\bbV\lmk \mathfrak S_{\mfp_1,\mfp_2,+1}^{(N)} \rmk$.
\end{description}
\end{rem}

\begin{lem}\label{line}
Consider Setting \ref{kame}.
Given $\bm m\in C_{\mathfrak {OE}_{\mfp_1,\mfp_2}^{(N)}
}$, 
$ \bm k\in C_{\mathbb E\lmk \mathfrak S_{\mfp_1,\mfp_2,+1}^{(N)} \rmk}$ and $g\in G$,
there exists a unique $\bm o:=\bm o\lmk\bm m, \bm k, g\rmk\in C_{{{\hln N}^{(n_0,m_0)}}}$ such that
\begin{align}
\begin{split}
\bm o\vert_{\mathfrak {OE}_{\mfp_1,\mfp_2}^{(N)}}=\bm m,\quad
\bm o\vert_{{\sppe}}=\bm k,\quad
\bm o\vert_{e_0^{\mfp_1,\mfp_2, (N)}}=g.
\end{split}
\end{align}
Furthermore, if 
\begin{align}\label{usagi}
\left.\Psi^{(x,y)}(\bm k)\right\vert_{\partial {\mathbb V\lmk \mathfrak S_{\mfp_1,\mfp_2,+1}^{(N)} \rmk}}
=\left.\Psi^{(x,y)}(\bm h)\right\vert_{\partial {\mathbb V\lmk \mathfrak S_{\mfp_1,\mfp_2,+1}^{(N)} \rmk}}
\end{align}
for 
$ \bm k, \bm h\in C_{\mathbb E\lmk \mathfrak S_{\mfp_1,\mfp_2,+1}^{(N)} \rmk}$,
then 
\begin{align}\label{oic}
\left. \bm o\lmk\bm m, \bm k, g\rmk\right\vert_{\mathfrak {BE}_{\mfp_1,\mfp_2}^{(N)}}
=\left. \bm o\lmk\bm m, \bm h, g\rmk\right\vert_{\mathfrak {BE}_{\mfp_1,\mfp_2}^{(N)}}
\end{align}
for any $\bm m
\in C_{\mathfrak {OE}_{\mfp_1,\mfp_2}^{(N)}
}
$.
\end{lem}
\begin{rem}
Recall Definition \ref{kame} for (\ref{usagi}).
\end{rem}
\begin{proof}
\underline{Existence}\\
Fix $\bm m\in C_{\mathfrak {OE}_{\mfp_1,\mfp_2}^{(N)}}
$, 
$ \bm k\in C_{\mathbb E\lmk \mathfrak S_{\mfp_1,\mfp_2,+1}^{(N)} \rmk
}$ and $g\in G$
labelling edges in $\mathfrak {OE}_{\mfp_1,\mfp_2}^{(N)}
$, $\mathbb E\lmk \mathfrak S_{\mfp_1,\mfp_2,+1}^{(N)} \rmk
$,
and $e_0^{\mfp_1,\mfp_2, (N)}$, respectively.
With this confguration, the admissibility condition holds on any squares in 
$\mathfrak {OE}_{\mfp_1,\mfp_2}^{(N)}
$ or  $\mathbb E\lmk \mathfrak S_{\mfp_1,\mfp_2,+1}^{(N)} \rmk
$.
Note that the squares in ${{\hln N}^{(n_0,m_0)}}$ which are not in 
$\mathfrak {OE}_{\mfp_1,\mfp_2}^{(N)}$ nor in  $\mathbb E\lmk \mathfrak S_{\mfp_1,\mfp_2,+1}^{(N)} \rmk$
are exactly $S_i^{\mfp_1,\mfp_2, (N)}$, $i=1,\ldots, {M_{\mfp_1,\mfp_2, (N)}}$.
The only edges in ${{\hln N}^{(n_0,m_0)}}$ which are not yet labeled are
$\mathfrak {BE}_{\mfp_1,\mfp_2}^{(N)}\setminus\{e_0^{\mfp_1,\mfp_2, (N)}\}=\{ e_i^{\mfp_1,\mfp_2, (N)}\}_{i=1}^{{M_{\mfp_1,\mfp_2, (N)}}}$.
We have to decide labels on these edges in a way that the admissibility condition holds on any of
$S_i^{\mfp_1,\mfp_2, (N)}$, $i=1,\ldots, {M_{\mfp_1,\mfp_2, (N)}}$.
But the admissibility conditions on $S_i^{\mfp_1,\mfp_2, (N)}$, $i=1,\ldots, {M_{\mfp_1,\mfp_2, (N)}}$
determines the labels on these edges
$\mathfrak {BE}_{\mfp_1,\mfp_2}^{(N)}\setminus\{e_0^{\mfp_1,\mfp_2, (N)}\}=\{ e_i^{\mfp_1,\mfp_2, (N)}\}_{i=1}^{{M_{\mfp_1,\mfp_2, (N)}}}$ uniquely : on the square
 $S_1^{\mfp_1,\mfp_2, (N)}$ two 
edges of $S_1^{\mfp_1,\mfp_2, (N)}$ are in 
$ {\mathbb E\lmk \mathfrak S_{\mfp_1,\mfp_2,+1}^{(N)} \rmk}\cup  \mathfrak {OE}_{\mfp_1,\mfp_2}^{(N)}
$.
Therefore, their label is already fixed by $\bm m$ and $\bm k$.
One of the rest edge is $e_0^{\mfp_1,\mfp_2, (N)}$, whose label is already fixed as $g$.
The label of the last edge $e_1^{\mfp_1,\mfp_2, (N)}$ is then determined uniquely
by the admissibility condition.
Next proceed to $S_2^{\mfp_1,\mfp_2, (N)}$.
Two 
edges of $S_2^{\mfp_1,\mfp_2, (N)}$ are in 
$ {\mathbb E\lmk \mathfrak S_{\mfp_1,\mfp_2,+1}^{(N)} \rmk}\cup  \mathfrak {OE}_{\mfp_1,\mfp_2}^{(N)}
$,
the other edges are $e_1^{\mfp_1,\mfp_2, (N)},e_2^{\mfp_1,\mfp_2, (N)}$.
Hence the first three of them have already fixed labels.
The admissibility condition determines the label of
$e_2^{\mfp_1,\mfp_2, (N)}$ uniquely.
Continuing this, we obtain the labels on 
$\mathfrak {BE}_{\mfp_1,\mfp_2}^{(N)}\setminus\{e_0^{\mfp_1,\mfp_2, (N)}\}=\{ e_i^{\mfp_1,\mfp_2, (N)}\}_{i=1}^{{M_{\mfp_1,\mfp_2, (N)}}}$
satisfying the admissibility condition on $S_i^{\mfp_1,\mfp_2, (N)}$, $i=1,\ldots, {M_{\mfp_1,\mfp_2, (N)}}$.
\\
\underline{Uniqueness and the last statement}\\
 Suppose that $ \bm k, \bm h\in C_{\mathbb E\lmk \mathfrak S_{\mfp_1,\mfp_2,+1}^{(N)} \rmk
}$ satisfy
\begin{align}\label{hkv}
\left.\Psi^{(x,y)}(\bm k)\right\vert_{\partial {\mathbb V\lmk \mathfrak S_{\mfp_1,\mfp_2,+1}^{(N)} \rmk}}
=\left.\Psi^{(x,y)}(\bm h)\right\vert_{\partial {\mathbb V\lmk \mathfrak S_{\mfp_1,\mfp_2,+1}^{(N)} \rmk}}
\end{align}
and $\bm o, \tilde{\bm o}\in C_{{{\hln N}^{(n_0,m_0)}}}$ satisfy
\begin{align}
\begin{split}
\bm o\vert_{\mathfrak {OE}_{\mfp_1,\mfp_2}^{(N)}
}=\bm m,\quad
\bm o\vert_{{\sppe}}=\bm k,\quad
\bm o\vert_{e_0^{\mfp_1,\mfp_2, (N)}}=g,\\
\tilde{\bm o}\vert_{\mathfrak {OE}_{\mfp_1,\mfp_2}^{(N)}
}=\bm m,\quad
\tilde{\bm o}\vert_{{\sppe}}=\bm h,\quad
\tilde{\bm o}\vert_{e_0^{\mfp_1,\mfp_2, (N)}}=g.
\end{split}
\end{align}
By the above argument, from the admissibility, the label of $\mathfrak {BE}_{\mfp_1,\mfp_2}^{(N)}\setminus\{e_0^{\mfp_1,\mfp_2, (N)}\}=\{ e_i^{\mfp_1,\mfp_2, (N)}\}_{i=1}^{{M_{\mfp_1,\mfp_2, (N)}}}$ for $\bm o$, $\tilde{\bm o}$
are determined by this condition uniquely.
We claim 
\begin{align}\label{eimc}
\bm o \vert_{e_i^{\mfp_1,\mfp_2, (N)}}=\tilde{\bm o} \vert_{e_i^{\mfp_1,\mfp_2, (N)}},\quad i=0,\ldots, {M_{\mfp_1,\mfp_2, (N)}}
\end{align}
corresponding to (\ref{oic}).
From the construction of $S_i^{\mfp_1,\mfp_2, (N)}$, $e_i^{\mfp_1,\mfp_2, (N)}$, we note the following.
\begin{description}
\item[(i)] Let $v, v'\in \partial{\mathbb V\lmk \mathfrak S_{\mfp_1,\mfp_2,+1}^{(N)} \rmk}$
and let $e\in \hat\bbZ^2$ be the edge $e=v-v'$.
Because of (\ref{hkv}),  the label of $e$ in $\bm o$, $\tilde{\bm o}$ coincides:
\begin{align}
\bm o\vert_{e}=
\lmk \left.\Psi^{(x,y)}(\bm k)\right\vert_{v}\rmk^{-1}\left.\Psi^{(x,y)}(\bm k)\right\vert_{v'}
=\lmk \left.\Psi^{(x,y)}(\bm h)\right\vert_{v}\rmk^{-1}\left.\Psi^{(x,y)}(\bm h)\right\vert_{v'}
 =\tilde{\bm o}\vert_e.
\end{align}
\item[(ii)]$\bm o\vert_{\mathfrak {OE}_{\mfp_1,\mfp_2}^{(N)}
}=\bm m=\tilde{\bm o}\vert_{\mathfrak {OE}_{\mfp_1,\mfp_2}^{(N)}
}$.
\item[(iii)] If $e_i^{\mfp_1,\mfp_2, (N)}$ and $e_{i-1}^{\mfp_1,\mfp_2, (N)}$ are parallel,
then the edges of  the square  $S_i^{\mfp_1,\mfp_2, (N)}$ are
$e_i^{\mfp_1,\mfp_2, (N)}$, $e_{i-1}^{\mfp_1,\mfp_2, (N)}$, and
one edge in ${\mathbb E\lmk \mathfrak S_{\mfp_1,\mfp_2,+1}^{(N)} \rmk}$, 
(with vertices in $\partial\bbV\lmk \mathfrak S_{\mfp_1,\mfp_2,+1}^{(N)} \rmk$) one edge in $\mathfrak {OE}_{\mfp_1,\mfp_2}^{(N)}
$.
From (i) and (ii), it means all the edges in $S_i^{\mfp_1,\mfp_2, (N)}$ except for $e_i^{\mfp_1,\mfp_2, (N)}$, $e_{i-1}^{\mfp_1,\mfp_2, (N)}$
have the same label for $\tilde {\bm o}$ and $\bm o$.
By the admissibility condition on $S_i^{\mfp_1,\mfp_2, (N)}$, it means that if 
$\bm o\vert_{e_{i-1}^{\mfp_1,\mfp_2, (N)}}=\tilde {\bm o}\vert_{e_{i-1}^{\mfp_1,\mfp_2, (N)}}$
then we have $\bm o\vert_{e_i^{\mfp_1,\mfp_2, (N)}}=\tilde {\bm o}\vert_{e_i^{\mfp_1,\mfp_2, (N)}}$.
\item[(iv)]
If $e_i^{\mfp_1,\mfp_2, (N)}$ and $e_{i-1}^{\mfp_1,\mfp_2, (N)}$ shares a point in $\partial{\mathbb V\lmk \mathfrak S_{\mfp_1,\mfp_2,+1}^{(N)} \rmk}$,
the edges in $S_i^{\mfp_1,\mfp_2, (N)}$ other than $e_i^{\mfp_1,\mfp_2, (N)}$, $e_{i-1}^{\mfp_1,\mfp_2, (N)}$ 
belong to $\mathfrak {OE}_{\mfp_1,\mfp_2}^{(N)}
$.
From (ii), it means all the edges in $S_i^{\mfp_1,\mfp_2, (N)}$ except for $e_i^{\mfp_1,\mfp_2, (N)}$, $e_{i-1}^{\mfp_1,\mfp_2, (N)}$
have the same label for $\tilde {\bm o}$ and $\bm o$.
By the admissibility condition, it means that if 
$\bm o\vert_{e_{i-1}^{\mfp_1,\mfp_2, (N)}}=\tilde {\bm o}\vert_{e_{i-1}^{\mfp_1,\mfp_2, (N)}}$
then we have $\bm o\vert_{e_i^{\mfp_1,\mfp_2, (N)}}=\tilde {\bm o}\vert_{e_i^{\mfp_1,\mfp_2, (N)}}$.
\item[(v)]
If $e_i^{\mfp_1,\mfp_2, (N)}$ and $e_{i-1}^{\mfp_1,\mfp_2, (N)}$ shares a point in 
${\Lambda_N^{(n_0,m_0)}}\setminus {\mathbb V\lmk \mathfrak S_{\mfp_1,\mfp_2,+1}^{(N)} \rmk}
$,
the edges $e,\tilde e$ in $S_i^{\mfp_1,\mfp_2, (N)}$ other than $e_i^{\mfp_1,\mfp_2, (N)}$, $e_{i-1}^{\mfp_1,\mfp_2, (N)}$ 
belong to ${{\mathbb E\lmk \mathfrak S_{\mfp_1,\mfp_2,+1}^{(N)} \rmk}}$.
We may assume $e$ and $e_{i-1}^{\mfp_1,\mfp_2, (N)}$ share a point 
$v\in \partial \mathbb V\lmk \mathfrak S_{\mfp_1,\mfp_2,+1}^{(N)} \rmk$,
$\tilde e$ and $e_{i}^{\mfp_1,\mfp_2, (N)}$ share a point 
$v'\in \partial \mathbb V\lmk \mathfrak S_{\mfp_1,\mfp_2,+1}^{(N)} \rmk$.
Because of (\ref{hkv}),
we have
\begin{align}
\begin{split}
&\lmk \lmk \bm o\rmk_{{e_{i-1}^{\mfp_1,\mfp_2, (N)}}} \rmk^{\sigma_{i-1}}
\lmk \lmk \bm o\rmk_{{e_{i}^{\mfp_1,\mfp_2, (N)}}} \rmk^{\sigma_{i}}
=\lmk \left.\Psi^{(x,y)}(\bm k)\right\vert_{v}\rmk^{-1}
\left.\Psi^{(x,y)}(\bm k)\right\vert_{v''}\\
&=\lmk \left.\Psi^{(x,y)}(\bm h)\right\vert_{v}\rmk^{-1}
\left.\Psi^{(x,y)}(\bm h)\right\vert_{v''}
=\lmk \lmk \tilde{\bm o}\rmk_{{e_{i-1}^{\mfp_1,\mfp_2, (N)}}} \rmk^{\sigma_{i-1}}
\lmk \lmk \tilde{\bm o}\rmk_{{e_{i}^{\mfp_1,\mfp_2, (N)}}} \rmk^{\sigma_{i}},
\end{split}
\end{align}
with some $\sigma_{i-1},\sigma_i=\pm 1$.
By the admissibility condition, it means that if 
$\bm o\vert_{e_{i-1}^{\mfp_1,\mfp_2, (N)}}=\tilde {\bm o}\vert_{e_{i-1}^{\mfp_1,\mfp_2, (N)}}$
then we have $\bm o\vert_{e_i^{\mfp_1,\mfp_2, (N)}}=\tilde {\bm o}\vert_{e_i^{\mfp_1,\mfp_2, (N)}}$.
\end{description}
We have $\bm o \vert_{e_0^{\mfp_1,\mfp_2, (N)}}=g=\tilde{\bm o} \vert_{e_0^{\mfp_1,\mfp_2, (N)}}$.
Then with the above observation, we have
$\bm o\vert_{e_i^{\mfp_1,\mfp_2, (N)}}=\tilde {\bm o}\vert_{e_i^{\mfp_1,\mfp_2, (N)}}$
for all $i=0,\ldots,{M_{\mfp_1,\mfp_2, (N)}}$ inductively, proving the claim (\ref{eimc}).
\end{proof}

\begin{lem} \label{gorira}
Consider Setting \ref{kame}.
For $\bm h, \bm k\in G^{\mathbb E\lmk \mathfrak S_{\mfp_1,\mfp_2,+1}^{(N)} \rmk}$
the followings are equivalent.
\begin{description}
\item[(i)]
For any 
$\bm m\in C_{\mathfrak {OE}_{\mfp_1,\mfp_2}^{(N)}
}$, 
and $g_0\in G^{e_0^{\mfp_1,\mfp_2, (N)}}$, 
there exists a unique 
$\lmk g_{i}\rmk_{i=1,\ldots,{M_{\mfp_1,\mfp_2, (N)}}}\in G^{\{ e_i^{\mfp_1,\mfp_2, (N)}\}_{i=1}^{{M_{\mfp_1,\mfp_2, (N)}}}}$, 
such that 
\begin{align}\label{chrip}
\begin{split}
\lmk \bm k,\bm m,  \bm g\rmk,\quad 
\lmk \bm h,\bm m,  \bm g\rmk
\in C_{{{\hln N}^{(n_0,m_0)}}},
\end{split}
\end{align}
and
\begin{align}\label{bara}
\begin{split}
\lmk
\bm k,\bm m,  \bm g
\rmk_{\partial{{{\hln N}^{(n_0,m_0)}}} }
=\lmk \bm h,\bm m, \bm g\rmk_{\partial{{{\hln N}^{(n_0,m_0)}}}}.
\end{split}
\end{align}
Here, we set $\bm g:=(g_{i})_{i=0}^{{M_{\mfp_1,\mfp_2, (N)}}}\in  G^{\mathfrak {BE}_{\mfp_1,\mfp_2}^{(N)}}$.
\item[(ii)] There exist $\bm m\in C_{\mathfrak {OE}_{\mfp_1,\mfp_2}^{(N)}
}$, and $g_0\in G^{e_0^{\mfp_1,\mfp_2, (N)}}$, 
which allow unique 
$\lmk g_{i}\rmk_{i=1,\ldots,{M_{\mfp_1,\mfp_2, (N)}}}\in G^{\{ e_i^{\mfp_1,\mfp_2, (N)}\}_{i=1}^{{M_{\mfp_1,\mfp_2, (N)}}}}$, 
such that 
\begin{align}\label{rose}
\begin{split}
\lmk \bm k,\bm m,  \bm g\rmk,\quad 
\lmk \bm h,\bm m,  \bm g\rmk
\in C_{{{\hln N}^{(n_0,m_0)}}},
\end{split}
\end{align}
and
\begin{align}\label{hamu}
\begin{split}
\lmk
\bm k,\bm m,  \bm g
\rmk_{\partial{{{\hln N}^{(n_0,m_0)}}} }
=\lmk \bm h,\bm m, \bm g\rmk_{\partial{{{\hln N}^{(n_0,m_0)}}}}.
\end{split}
\end{align}
Here, we set $\bm g:=(g_{i})_{i=0}^{{M_{\mfp_1,\mfp_2, (N)}}}\in  
G^{\mathfrak {BE}_{\mfp_1,\mfp_2}^{(N)}}$.
\item[(iii)]
$ \bm k, 
\bm h\in C_{\mathbb E\lmk \mathfrak S_{\mfp_1,\mfp_2,+1}^{(N)} \rmk}$
and
$
\Psi^{(x,y)}(\bm h)\vert_{{\partial {\mathbb V\lmk \mathfrak S_{\mfp_1,\mfp_2,+1}^{(N)} \rmk}}}=
\Psi^{(x,y)}(\bm k)\vert_{{\partial {\mathbb V\lmk \mathfrak S_{\mfp_1,\mfp_2,+1}^{(N)} \rmk}}}
$.
\end{description}
\end{lem}
\begin{rem}
In (iii), $\Psi^{(x,y)}(\bm h)$,
$
\Psi^{(x,y)}(\bm k)$
are well-defined because $ \bm k, 
\bm h\in C_{\mathbb E\lmk \mathfrak S_{\mfp_1,\mfp_2,+1}^{(N)} \rmk}$.
\end{rem}
\begin{proof}
(i)$\Rightarrow $(ii) : Note that $C_{\mathfrak {OE}_{\mfp_1,\mfp_2}^{(N)}}$ is non-empty.
For example, we may set all the labels in $\mathfrak {OE}_{\mfp_1,\mfp_2}^{(N)}$ to be the unit of
the group $G$.
Hence (i) implies (ii).\\
(ii)$\Rightarrow $(iii) : The first part  
$ \bm k, 
\bm h\in C_{{\mathbb E\lmk \mathfrak S_{\mfp_1,\mfp_2,+1}^{(N)} \rmk}}$ is trivial.
We would like to show 
$
\Psi^{(x,y)}(\bm h)\vert_{{{\partial {\mathbb V\lmk \mathfrak S_{\mfp_1,\mfp_2,+1}^{(N)} \rmk}}}}=
\Psi^{(x,y)}(\bm k)\vert_{{{\partial {\mathbb V\lmk \mathfrak S_{\mfp_1,\mfp_2,+1}^{(N)} \rmk}}}}
$.
For any $v\in \partial \mathbb V\lmk \mathfrak S_{\mfp_1,\mfp_2,+1}^{(N)} \rmk\setminus\partial{\Lambda_N^{(n_0,m_0)}}$, as in (E) Remark \ref{obs}
there exists a path $\mfp_v$ in ${{\hln N}^{(n_0,m_0)}}\setminus \mathbb E\lmk \mathfrak S_{\mfp_1,\mfp_2,+1}^{(N)} \rmk$
with origin $v$ and target $v_0=(x,y)$.
Note because of $
{\lmk \bm k,{\bm m, \bm g}\rmk, 
\lmk  \bm h,{\bm m, \bm g}\rmk\in C_{\hln {N}^{(n_0,m_0)}}}
$
that 
\begin{align}\label{chk}
\begin{split}
\Psi^{(x,y)}(\bm h)_{v}=\Psi_{\mfp_v^{-1}}\lmk  \bm h,{\bm m, \bm g}\rmk,\\
\Psi^{(x,y)}(\bm k)_{v}=\Psi_{\mfp_v^{-1}}\lmk  \bm k,{\bm m, \bm g}\rmk.\\
\end{split}
\end{align}
Because the path $\mfp_v$ is in ${{\hln N}^{(n_0,m_0)}}\setminus \mathbb E\lmk \mathfrak S_{\mfp_1,\mfp_2,+1}^{(N)} \rmk$,
we have
\begin{align}
\lmk  \bm h,{\bm m, \bm g}\rmk_{\mfp_v}=\lmk \bm m, \bm g\rmk_{\mfp_v}=\lmk  \bm k,{\bm m, \bm g}\rmk_{\mfp_v}.
\end{align}
Hence from (\ref{chk}), we obtain
\begin{align}
\begin{split}
\Psi^{(x,y)}(\bm h)_{v}=\Psi^{(x,y)}(\bm k)_{v},\quad v\in \partial \mathbb V\lmk \mathfrak S_{\mfp_1,\mfp_2,+1}^{(N)} \rmk\setminus\partial{\Lambda_N^{(n_0,m_0)}},\quad .
\end{split}
\end{align}
For $v\in\partial{\Lambda_N^{(n_0,m_0)}}$, as in (F) Remark \ref{obs}, there is a path $\mfp_v^1$ along $\partial{{\hln N}^{(n_0,m_0)}}$, 
with origin $v$ and target $(Nn_0,-Nm_0)$.
There is a path $\mfp_0$ in ${{\hln N}^{(n_0,m_0)}}\setminus \mathbb E\lmk \mathfrak S_{\mfp_1,\mfp_2,+1}^{(N)} \rmk$
with origin $(Nn_0,-Nm_0)$ and target $v_0$.
On the path $\mfp_v^1$, 
the value of $(\bm k,{\bm m, \bm g})$ and $ (\bm h,{\bm m, \bm g})$ are the same because of
(\ref{hamu}).
Because the path $\mfp_0$ is in ${{\hln N}^{(n_0,m_0)}}\setminus \mathbb E\lmk \mathfrak S_{\mfp_1,\mfp_2,+1}^{(N)} \rmk$
we again have
\begin{align}
\lmk  \bm h,{\bm m, \bm g}\rmk_{\mfp_0}=\lmk \bm m, \bm g\rmk_{\mfp_0}=\lmk  \bm k,{\bm m, \bm g}\rmk_{\mfp_0}.
\end{align}
Hence from (\ref{rose}) we get
\begin{align}
\begin{split}
\Psi^{(x,y)}(\bm h)_{v}
=\Psi_{\lmk \mfp_v^1\mfp_0\rmk^{-1}}\lmk  \bm h,{\bm m, \bm g}\rmk
=\Psi_{\lmk \mfp_v^1\mfp_0\rmk^{-1}}\lmk  \bm k,{\bm m, \bm g}\rmk
=\Psi^{(x,y)}(\bm k)_{v},\quad v\in\partial{\Lambda_N^{(n_0,m_0)}}.
\end{split}
\end{align}
This completes the proof.

\noindent
(iii) $\Rightarrow$ (i) : Let  $\bm o\lmk\bm m, \bm k, g_0\rmk, \bm o\lmk\bm m, \bm h, g_0\rmk\in C_{{{\hln N}^{(n_0,m_0)}}}$
be the configurations obtained in  Lemma \ref{line} from 
$\bm m\in C_{\mathfrak {OE}_{\mfp_1,\mfp_2}^{(N)}}$, 
 $g_0\in G^{e_0^{\mfp_1,\mfp_2,(N)}}$, $ \bm k,\bm h\in C_{\mathbb E\lmk \mathfrak S_{\mfp_1,\mfp_2,+1}^{(N)} \rmk}$
 respectively.
By Lemma \ref{line}, (iii) implies
\begin{align}\label{oici}
\bm g:=\left. \bm o\lmk\bm m, \bm k, g_0\rmk\right\vert_{\mathfrak {BE}_{\mfp_1,\mfp_2}^{(N)}}
=\left. \bm o\lmk\bm m, \bm h, g_0\rmk\right\vert_{\mathfrak {BE}_{\mfp_1,\mfp_2}^{(N)}}
\in G^{\mathfrak {BE}_{\mfp_1,\mfp_2}^{(N)}},
\end{align}
and 
\begin{align}
\begin{split}
\lmk \bm k,\bm m,  \bm g\rmk=\bm o\lmk\bm m, \bm k, g_0\rmk,\quad
\lmk \bm h,\bm m,  \bm g\rmk=\bm o\lmk\bm m, \bm h, g_0\rmk\in C_{{{\hln N}^{(n_0,m_0)}}}.
\end{split}
\end{align}
By Lemma \ref{line},  this 
$\lmk g_{i}\rmk_{i=1,\ldots,{M_{\mfp_1,\mfp_2, (N)}}}\in G^{\{ e_i^{\mfp_1,\mfp_2, (N)}\}_{i=1}^{{M_{\mfp_1,\mfp_2, (N)}}}}$
is the unique configuration 
satisfying (\ref{chrip}).
Furthermore for any $e\in \partial{{{\hln N}^{(n_0,m_0)}}}\cap { {\mathbb E\lmk \mathfrak S_{\mfp_1,\mfp_2,+1}^{(N)} \rmk}}$,
$e$ is of the form $e=v-v'$ with $v, v'\in {\partial {\mathbb V\lmk \mathfrak S_{\mfp_1,\mfp_2,+1}^{(N)} \rmk}}$.
Therefore, by (iii), we have
\begin{align}
\begin{split}
\lmk \bm k,\bm m,  \bm g\rmk_{e}
=\lmk \Psi^{(x,y)}(\bm k)\rmk_{v}^{-1}
\Psi^{(x,y)}(\bm k)_{v'}
=\lmk \Psi^{(x,y)}(\bm h)\rmk_{v}^{-1}
\Psi^{(x,y)}(\bm h)_{v'}
=\lmk \bm h,\bm m,  \bm g\rmk_{e}.
\end{split}
\end{align}
For  any $e\in \partial{{{\hln N}^{(n_0,m_0)}}}\setminus { {\mathbb E\lmk \mathfrak S_{\mfp_1,\mfp_2,+1}^{(N)} \rmk}}$,
we have $e\in \mathfrak {OE}_{\mfp_1,\mfp_2}^{(N)}\cup \mathfrak {BE}_{\mfp_1,\mfp_2}^{(N)}$.
But we have 
\begin{align}
\begin{split}
\lmk \bm k,\bm m,  \bm g\rmk\vert_{\mathfrak {BE}_{\mfp_1,\mfp_2}^{(N)}}=\bm o\lmk\bm m, \bm k, g_0\rmk\vert_{\mathfrak {BE}_{\mfp_1,\mfp_2}^{(N)}}
=\bm o\lmk\bm m, \bm h, g_0\rmk\vert_{\mathfrak {BE}_{\mfp_1,\mfp_2}^{(N)}}=\lmk \bm h,\bm m,  \bm g\rmk\vert_{\mathfrak {BE}_{\mfp_1,\mfp_2}^{(N)}},
\end{split}
\end{align}
from (\ref{oici}), and
\begin{align}
\begin{split}
\lmk \bm k,\bm m,  \bm g\rmk\vert_{\mathfrak {OE}_{\mfp_1,\mfp_2}^{(N)}}=\bm o\lmk\bm m, \bm k, g_0\rmk\vert_{\mathfrak {OE}_{\mfp_1,\mfp_2}^{(N)}}
=\bm m\vert_{\mathfrak {OE}_{\mfp_1,\mfp_2}^{(N)}}
=\bm o\lmk\bm m, \bm h, g_0\rmk\vert_{\mathfrak {OE}_{\mfp_1,\mfp_2}^{(N)}}=\lmk \bm h,\bm m,  \bm g\rmk\vert_{\mathfrak {OE}_{\mfp_1,\mfp_2}^{(N)}}.
\end{split}
\end{align}
Hence we have
\begin{align}
\begin{split}
\lmk
\bm k,\bm m,  \bm g
\rmk_e
=\lmk \bm h,\bm m, \bm g\rmk_e,\quad e
\in \partial{{{\hln N}^{(n_0,m_0)}}}\setminus { {\mathbb E\lmk \mathfrak S_{\mfp_1,\mfp_2,+1}^{(N)} \rmk}}.
\end{split}
\end{align}
Hence we obtain
 (\ref{bara}) 
$
\lmk
\bm k,\bm m,  \bm g
\rmk_{\partial{{{\hln N}^{(n_0,m_0)}}} }
=\lmk \bm h,\bm m, \bm g\rmk_{\partial{{{\hln N}^{(n_0,m_0)}}}}
$.

\end{proof}
Consider Setting \ref{kame}.
Recall the loop $\mathfrak c_{\mfp_1,\mfp_2,+1}^{(N)}=\mfp_1^{(N)}\mathfrak l_{\mfp_1,\mfp_2,+1}^{(N)}\lmk \mfp_2^{(N)}\rmk^{-1}$
from Definition \ref{ws}.
Because $\mfp_1$ and $\mfp_2$ are well-separated,
the $\mathfrak S_{\mfp_1,\mfp_2,+1}^{(N)}$, the squares inside of
the loop $\mathfrak c_{\mfp_1,\mfp_2,+1}^{(N)}$
satisfy the condition S.
Therefore, there are two distinct self-avoiding paths $\mfp_{{\mfp_1,\mfp_2,N}}^{(l)}$, 
$\mfp_{{\mfp_1,\mfp_2,N}}^{(r)}$ Definition \ref{ringo} corresponding to
$\mathfrak S_{\mfp_1,\mfp_2,+1}^{(N)} $.
Note that $\mathfrak c_{\mfp_1,\mfp_2,+1}^{(N)}=\lmk\mfp_{{\mfp_1,\mfp_2,N}}^{(l)} \rmk^{-1}\mfp_{{\mfp_1,\mfp_2,N}}^{(r)}$.
We recall from section \ref{acls} that
$\caE^{(4)}\lmk\mathfrak S_{\mfp_1,\mfp_2,+1}^{(N)}\rmk$ denotes 
the set of all edges in $\mathfrak c_{\mfp_1,\mfp_2,+1}^{(N)}$.
We denote by $\bbV\lmk\mathfrak c_{\mfp_1,\mfp_2,+1}^{(N)}\rmk$
the set of all vertices in $\mathfrak c_{\mfp_1,\mfp_2,+1}^{(N)}$.
We also set
\begin{align}
\begin{split}
&{J}_{{\mfp_1,\mfp_2},{N} }:=\bbV\lmk\mathfrak c_{\mfp_1,\mfp_2,+1}^{(N)}\rmk\setminus {{\partial {\mathbb V\lmk \mathfrak S_{\mfp_1,\mfp_2,+1}^{(N)} \rmk}}}.
\end{split}
\end{align}
Because $v_0\in  {{\partial {\mathbb V\lmk \mathfrak S_{\mfp_1,\mfp_2,+1}^{(N)} \rmk}}}$, we have
$v_0\notin {J}_{{\mfp_1,\mfp_2},{N} }$.
We consider the following set of configurations
\begin{align}
\begin{split}
&\mathcal{FB}_{{\mfp_1,\mfp_2,N}}
:=\left\{
\bm a\in G^{\caE^{(4)}\lmk\mathfrak S_{\mfp_1,\mfp_2,+1}^{(N)}\rmk}\mid
\Psi_{\mfp_{{\mfp_1,\mfp_2,N}}^{(l)}}(\bm a)=\Psi_{\mfp_{{\mfp_1,\mfp_2,N}}^{(r)}}(\bm a)
\right\},\\
&\mathcal{PT}_{{\mfp_1,\mfp_2,N}}:=
\left\{
\bm t\in G^{\partial \mathbb V\lmk \mathfrak S_{\mfp_1,\mfp_2,+1}^{(N)} \rmk}\mid
t_{v_0}=e
\right\}.
\end{split}
\end{align}
Note for any $\bm a\in \mathcal{FB}_{{\mfp_1,\mfp_2,N}}$ and  
any $w\in \bbV\lmk\mathfrak c_{\mfp_1,\mfp_2,+1}^{(N)}\rmk$,
the value $\Psi_{\tilde{\mfq}_w}(\bm a)$ is independent of the choice of
the path $\tilde{\mfq}_w$ in $\caE^{(4)}\lmk\mathfrak S_{\mfp_1,\mfp_2,+1}^{(N)}\rmk$  with origin
 $v_0$ and target $w$.
 We denote this value  by 
 $\lmk \Psi_{\caE^{(4)}\lmk\mathfrak S_{\mfp_1,\mfp_2,+1}^{(N)}\rmk}^{(x, y)}(\bm a)\rmk_w$
 and set
 \begin{align}\label{wednesday}
 \begin{split}
  \Psi_{\caE^{(4)}\lmk\mathfrak S_{\mfp_1,\mfp_2,+1}^{(N)}\rmk}^{(x, y)}(\bm a):=
 \lmk \lmk \Psi_{\caE^{(4)}\lmk\mathfrak S_{\mfp_1,\mfp_2,+1}^{(N)}\rmk}^{(x, y)}(\bm a)\rmk_w\rmk_{w\in \bbV\lmk\mathfrak c_{\mfp_1,\mfp_2,+1}^{(N)}\rmk
 }.
 \end{split}
 \end{align}
The restriction of this to ${{\partial {\mathbb V\lmk \mathfrak S_{\mfp_1,\mfp_2,+1}^{(N)} \rmk}}}$
is denoted by
\begin{align}
 \Psi_{{{\partial {\mathbb V\lmk \mathfrak S_{\mfp_1,\mfp_2,+1}^{(N)} \rmk}}}}^{(x, y)}(\bm a)
:= \left.
 \Psi_{\caE^{(4)}\lmk\mathfrak S_{\mfp_1,\mfp_2,+1}^{(N)}\rmk}^{(x, y)}(\bm a)
 \right\vert_{{{\partial {\mathbb V\lmk \mathfrak S_{\mfp_1,\mfp_2,+1}^{(N)} \rmk}}}}
\end{align}

\begin{lem}\label{risu}
Consider Setting \ref{kame}.
Then there is a bijection 
\begin{align}\label{amap}
\begin{split}
\Xi_{{\mfp_1,\mfp_2,N}}:  \mathcal{PT}_{{\mfp_1,\mfp_2,N}}
\times G^{\lv{J}_{{\mfp_1,\mfp_2,N}}\rv }\to 
 \mathcal{FB}_{{\mfp_1,\mfp_2,N}}
\end{split}
\end{align}
such that
\begin{align}\label{bdeq}
 \begin{split}
  \Psi_{{{\partial {\mathbb V\lmk \mathfrak S_{\mfp_1,\mfp_2,+1}^{(N)} \rmk}}}}^{(x, y)}\lmk\Xi_{{\mfp_1,\mfp_2,N}}(\bm t, \bm u)\rmk
=  \bm t ,\quad \lmk\bm t,\bm u\rmk\in \mathcal{PT}_{{\mfp_1,\mfp_2,N}}
\times G^{\lv{J}_{{\mfp_1,\mfp_2,N}}\rv }.
 \end{split}
 \end{align}

\end{lem}
\begin{rem}
If ${J}_{{\mfp_1,\mfp_2},{N} }=\emptyset$,
$\mathcal{PT}_{{\mfp_1,\mfp_2,N}}
\times G^{\lv{J}_{{\mfp_1,\mfp_2,N}}\rv }$
should be understood as $\mathcal{PT}_{{\mfp_1,\mfp_2,N}}$.

\end{rem}
\begin{proof}
If ${J}_{{\mfp_1,\mfp_2},{N} }\neq\emptyset$, then we fix some
$\check v\in {J}_{{\mfp_1,\mfp_2},{N} }$.
Let $\check f_1$, $\check f_2$ be the only two edges in 
$\caE^{(4)}\lmk\mathfrak S_{\mfp_1,\mfp_2,+1}^{(N)}\rmk$
 which have $\check v$ as their vertices.
 Let $\check v_1$ (resp. $\check v_2$ ) be the vertex of $\check f_1$
 (resp.  $\check f_2$) which is not $\check v$.
We may assume that as we proceed with respect to the direction given by the loop
$\mathfrak c_{\mfp_1,\mfp_2,+1}^{(N)}$, we encounter $\check v_1, \check v, \check v_2$ in succession. 
For $e\in \caE^{(4)}\lmk\mathfrak S_{\mfp_1,\mfp_2,+1}^{(N)}\rmk$, set $\sigma_{e}:=+1$ (resp. $\sigma_{e}:=-1$) if $\mathfrak c_{\mfp_1,\mfp_2,+1}^{(N)}$ goes through $e$
in the positive (resp. negative) direction.
We also denote by $\check \mfq$ the path of edges obtained from the loop
$\mathfrak c_{\mfp_1,\mfp_2,+1}^{(N)}$ by removing edges $\check f_1$, $\check f_2$.
The direction of $\check \mfq$ is inherited from that of $\mathfrak c_{\mfp_1,\mfp_2,+1}^{(N)}$.

For each $\bm z\in \mathcal{PT}_{{\mfp_1,\mfp_2,N}}\times G^{{J}_{{\mfp_1,\mfp_2},{N} }\setminus \{\check v\}}$
  define
\[
\bm s(\bm z):=(s_e(\bm z))_{e\in \check \mfq }\;\in
G^{\check \mfq}
\]
by
\begin{align}
\begin{split}
s_e(\bm z):=\lmk \lmk z_{v_{e,-\sigma_e}}\rmk^{-1}z_{v_{e,\sigma_e}}\rmk^{\sigma_e},\quad e\in \check \mfq.
\end{split}
\end{align}
Furthermore, for each $ \bm z\in  \mathcal{PT}_{{\mfp_1,\mfp_2,N}}\times G^{{J}_{{\mfp_1,\mfp_2},{N} }\setminus \{\check v\}}$ and
$u\in G^{\check f_1}$,
we set
 $\bm a(\bm z,u):=(a_e(\bm z, u))_{e\in \caE^{(4)}\lmk\mathfrak S_{\mfp_1,\mfp_2,+1}^{(N)}\rmk}\in G^{\caE^{(4)}\lmk\mathfrak S_{\mfp_1,\mfp_2,+1}^{(N)}\rmk}$ by
\begin{align}
\begin{split}
a_{e}(\bm z, u)
:=\left\{
\begin{gathered}
s_e(\bm z),\quad\text{if}\;\;e\in\check \mfq
,\\
u,\quad\text{if}\;\;e=\check f_1,\\
 \lmk
 \Psi_{\check \mfq}\lmk \bm s(\bm z)\rmk u^{\sigma_{\check f_1}}
 \rmk^{-\sigma_{\check f_2}}
,\quad\text{if}\;\;e=\check f_2
\end{gathered}
\right.,\quad\quad
e\in \caE^{(4)}\lmk\mathfrak S_{\mfp_1,\mfp_2,+1}^{(N)}\rmk.
\end{split}
\end{align}
(If ${J}_{{\mfp_1,\mfp_2},{N} }=\emptyset$, then  regard it as
$a_{e}(\bm z)=s_e(\bm z)$.)
We claim  $\bm a(\bm z, u)\in \mathcal{FB}_{{\mfp_1,\mfp_2,N}}$.
If ${J}_{{\mfp_1,\mfp_2},{N} }\neq\emptyset$, by the choice of 
$a_{\check f_2}(\bm z, u)$, we have 
\begin{align}
\begin{split}
\Psi_{\check \mfq}\lmk \bm s(\bm z)\rmk u^{\sigma_{\check f_1}}\cdot 
\lmk 
\lmk
 \Psi_{\check \mfq}\lmk \bm s(\bm z)\rmk u^{\sigma_{\check f_1}}
 \rmk^{-\sigma_{\check f_2}}
 \rmk^{\sigma_{\check f_2}}
=e.
\end{split}
\end{align}
Hence $\Psi_{\mfp_{{\mfp_1,\mfp_2,N}}^{(l)}}(\bm a(\bm z, u))=\Psi_{\mfp_{{\mfp_1,\mfp_2,N}}^{(r)}}(\bm a(\bm z, u))$.
If ${J}_{{\mfp_1,\mfp_2},{N} }=\emptyset$, by the definition of $s_e(\bm z)$,
we have
\begin{align}
\begin{split}
&\lmk \Psi_{\mfp_{{\mfp_1,\mfp_2,N}}^{(l)}}(\bm a(\bm z, u))\rmk^{-1}
\Psi_{\mfp_{{\mfp_1,\mfp_2,N}}^{(r)}}(\bm a(\bm z, u))
=\lmk s_{e_1}(\bm z)\rmk^{\sigma_{e_1}}
\lmk s_{e_2}(\bm z)\rmk^{\sigma_{e_2}}\cdots\\
&=\lmk \lmk \lmk z_{v_{e_1,-\sigma_{e_1}}}\rmk^{-1}z_{v_{e_1,\sigma_{e_1}}}\rmk^{\sigma_{e_1}}\rmk^{\sigma_{e_1}}
\lmk \lmk \lmk z_{v_{e_2,-\sigma_{e_2}}}\rmk^{-1}z_{v_{e_2,\sigma_{e_2}}}\rmk^{\sigma_{e_2}}\rmk^{\sigma_{e_2}}
\cdots \lmk \lmk \lmk z_{v_{e_L,-\sigma_{e_L}}}\rmk^{-1}z_{v_{e_L,\sigma_{e_L}}}\rmk^{\sigma_{e_L}}\rmk^{\sigma_{e_L}}\\
&= \lmk z_{v_{e_1,-\sigma_{e_1}}}\rmk^{-1}z_{v_{e_1,\sigma_{e_1}}}
 \lmk z_{v_{e_2,-\sigma_{e_2}}}\rmk^{-1}z_{v_{e_2,\sigma_{e_2}}}
 \cdots \lmk z_{v_{e_L,-\sigma_{e_L}}}\rmk^{-1}z_{v_{e_L,\sigma_{e_L}}}\\
& =\lmk z_{v_{e_1,-\sigma_{e_1}}}\rmk^{-1}z_{v_{e_L,\sigma_{e_L}}}
=e.
\end{split}
\end{align}
Here we labeled the edges in $\lmk\mfp_{{\mfp_1,\mfp_2,N}}^{(l)} \rmk^{-1}\mfp_{{\mfp_1,\mfp_2,N}}^{(r)}$ in order
as $\{e_i\}_{i=1}^L$.
Note that $v_{e_{i}, \sigma_i}=v_{e_{i+1}, -\sigma_{i+1}}$, $i=1,\ldots, L$
and $v_{e_{L}, \sigma_L}=v_{e_{1}, -\sigma_{1}}$ because  $\mathfrak c_{\mfp_1,\mfp_2,+1}^{(N)}$ is a loop.
Hence we prove the claim $\bm a(\bm z, u)\in \mathcal{FB}_{{\mfp_1,\mfp_2,N}}$
and obtain a map 
\begin{align}\label{amap}
\begin{split}
 \Xi_{{\mfp_1,\mfp_2,N}}:
 \mathcal{PT}_{{\mfp_1,\mfp_2,N}}\times G^{{J}_{{\mfp_1,\mfp_2},{N} }\setminus \{\check v\}}\times G^{\check f_1}\ni
  \lmk \bm t,\bm r, u\rmk\mapsto
\bm a(\bm t,\bm r, u)\in \mathcal{FB}_{{\mfp_1,\mfp_2,N}}.
\end{split}
\end{align}
We regard $G^{{J}_{{\mfp_1,\mfp_2},{N} }\setminus \{\check v\}}\times G^{\check f_1}$ as
 $ G^{\lv{J}_{{\mfp_1,\mfp_2,N}}\rv }$ and obtain the map (\ref{amap}).

Because 
$\Xi_{{\mfp_1,\mfp_2,N}}(\bm t,\bm r, u)=\bm a(\bm t,\bm r, u)\in \mathcal{FB}_{{\mfp_1,\mfp_2,N}} $,
from (\ref{wednesday}),
we can define
$ \Psi_{\caE^{(4)}\lmk\mathfrak S_{\mfp_1,\mfp_2,+1}^{(N)}\rmk}^{(x, y)}\lmk\Xi_{{\mfp_1,\mfp_2,N}}(\bm t,\bm r, u) \rmk$.
In particular, for any $w\in \bbV\lmk\mathfrak c_{\mfp_1,\mfp_2,+1}^{(N)}\rmk\setminus\{\check v\}$,
let $\mfq_w$ be the portion of $\mathfrak c_{\mfp_1,\mfp_2,+1}^{(N)}$ from $(x,y)$ to $w$
which does not go through $\check v$.
Then we have
\begin{align}\label{tokei}
\begin{split}
 \Psi_{\caE^{(4)}\lmk\mathfrak S_{\mfp_1,\mfp_2,+1}^{(N)}\rmk}^{(x, y)}\lmk\Xi_{{\mfp_1,\mfp_2,N}}(\bm t,\bm r, u) \rmk_w
 =\Psi_{\mfq_w}\lmk\Xi_{{\mfp_1,\mfp_2,N}}(\bm t,\bm r, u) \rmk_w
 =\lmk \lmk\bm t,\bm r\rmk_{v_0}\rmk^{-1}\lmk\bm t,\bm r\rmk_w=\lmk\bm t,\bm r\rmk_w,
\end{split}
\end{align}
by the definition of $\bm s(\lmk\bm t,\bm r\rmk)$ and 
the fact that $\lmk\bm t,\bm r\rmk_{v_0}=e$. In particular, (\ref{bdeq}) holds.

The map $ \Xi_{{\mfp_1,\mfp_2,N}}$
is an injection.
%
In fact, if $\bm a(\bm t^{(1)},\bm r^{(1)}, u^{(1)})=\bm a(\bm t^{(2)},\bm r^{(2)}, u^{(2)})$, then 
clearly $u^{(1)}=u^{(2)}$ and (\ref{tokei}) implies $\lmk \bm t^{(1)},\bm r^{(1)}\rmk
=\lmk \bm t^{(2)},\bm r^{(2)}\rmk$.

It is also a surjection.
In fact, for any $\bm a=(a_e)_{e\in \caE^{(4)}\lmk\mathfrak S_{\mfp_1,\mfp_2,+1}^{(N)}\rmk}\in \mathcal{FB}_{{\mfp_1,\mfp_2,N}}$, set 
\begin{align}\label{tus}
\begin{split}
&z_w:=\Psi_{\mfq_w}(\bm a),\quad w\in  \bbV\lmk\mathfrak c_{\mfp_1,\mfp_2,+1}^{(N)}\rmk\setminus\{\check v\}
,\\
&u:=a_{\check f_1}
\end{split}
\end{align}
with $\mfq_w$ above.
Then $z_{(x,y)}=e$ and we have $\bm z:=(z_w)_{ w\in  \bbV\lmk\mathfrak c_{\mfp_1,\mfp_2,+1}^{(N)}\rmk\setminus\{\check v\}}\in \mathcal{PT}_{{\mfp_1,\mfp_2,N}}\times G^{{J}_{{\mfp_1,\mfp_2},{N} }\setminus \{\check v\}}$.
For this $\bm z$,
for any $e\in\caE^{(4)}\lmk\mathfrak S_{\mfp_1,\mfp_2,+1}^{(N)}\rmk\setminus\{\check f_1,\check f_2\}$,
we have
\begin{align}\label{sora}
\begin{split}
s_e(\bm z)=\lmk \lmk z_{v_{e,-\sigma_e}}\rmk^{-1}z_{v_{e,\sigma_e}}\rmk^{\sigma_e}
=\lmk \lmk \Psi_{\mfq_{v_{e,-\sigma_e}}}(\bm a)\rmk^{-1} \Psi_{\mfq_{v_{e,\sigma_e}}}(\bm a)\rmk^{\sigma_e}
=a_e.
\end{split}
\end{align}
Furthermore, from (\ref{sora}), we have
\begin{align}
\begin{split}
&\lmk
 \Psi_{\check \mfq}\lmk \bm s(\bm z)\rmk u^{\sigma_{\check f_1}}
 \rmk^{-\sigma_{\check f_2}}
=
\lmk
\Psi_{\check \mfq}\lmk \bm a\rmk 
\lmk a_{\check f_1}\rmk^{\sigma_{\check f_1}}
 \rmk^{-\sigma_{\check f_2}}
=\lmk a_{\check f_2}^{-\sigma_{\check f_2}}\rmk^{-\sigma_{\check f_2}}=
a_{\check f_2}.
\end{split}
\end{align}
Here, we used $\Psi_{\mathfrak c_{\mfp_1,\mfp_2,+1}^{(N)}}\lmk\bm a\rmk=e$, from $\bm a\in \mathcal{FB}_{{\mfp_1,\mfp_2,N}}$.
Hence we obtain
$\Xi_{{\mfp_1,\mfp_2,N}}
(\bm z, u)
=\bm a
$,
proving the surjectivity.
\end{proof}

\begin{lem}\label{gcnt}
Consider Setting \ref{kame}.
For any $\bm t\in G^{\partial \mathbb V\lmk \mathfrak S_{\mfp_1,\mfp_2,+1}^{(N)} \rmk}$,
we have
\begin{align}
\begin{split}
\lv\left\{
\bm k\in C_{{\mathbb E\lmk \mathfrak S_{\mfp_1,\mfp_2,+1}^{(N)} \rmk}}\mid
\Psi^{(x,y)}(\bm k)\vert_{{{\partial {\mathbb V\lmk \mathfrak S_{\mfp_1,\mfp_2,+1}^{(N)} \rmk}}}}=\bm t
\right\}\rv=\lv G\rv^{\lv {J}_{{\mfp_1,\mfp_2},{N} }\rv}
|G|^{\lv \caE_{\mathfrak S_{\mfp_1,\mfp_2,+1}^{(N)}}^{(1)}\rv}
\end{split}
\end{align}
In particular, it
does not depend on the choice of
$\bm t\in G^{\partial \mathbb V\lmk \mathfrak S_{\mfp_1,\mfp_2,+1}^{(N)} \rmk}$.
\end{lem}
\begin{proof}
For any $\bm t\in \mathcal{PT}_{{\mfp_1,\mfp_2,N}}$,
set
\begin{align}
\begin{split}
\mathcal{FB}_{{\mfp_1,\mfp_2,N}}^{\bm t}
:=\left\{\bm a \in \mathcal{FB}_{{\mfp_1,\mfp_2,N}}\mid
\Psi_{{{\partial {\mathbb V\lmk \mathfrak S_{\mfp_1,\mfp_2,+1}^{(N)} \rmk}}}}^{(x, y)}(\bm a)=\bm t
\right\}.
\end{split}
\end{align}
Then we have
\begin{align}\label{alpaka}
\begin{split}
&\left\{
\bm k\in C_{{{\mathbb E\lmk \mathfrak S_{\mfp_1,\mfp_2,+1}^{(N)} \rmk}}}\mid
\Psi^{(x,y)}\lmk \bm k \rmk\vert_{\partial \mathbb V\lmk \mathfrak S_{\mfp_1,\mfp_2,+1}^{(N)} \rmk}=\bm t
\right\}\\
&=\dot{\bigcup}_{\substack{\bm a\in \mathcal{FB}_{{\mfp_1,\mfp_2,N}}
,\\
 }
 }
\left\{
\bm k\in C_{{{\mathbb E\lmk \mathfrak S_{\mfp_1,\mfp_2,+1}^{(N)} \rmk}}}\mid\;\;
\Psi^{(x,y)}\lmk \bm k \rmk\vert_{\caE^{(4)}\lmk\mathfrak S_{\mfp_1,\mfp_2,+1}^{(N)}\rmk}=\bm a,\;
\Psi^{(x,y)}\lmk \bm k \rmk\vert_{\partial \mathbb V\lmk \mathfrak S_{\mfp_1,\mfp_2,+1}^{(N)} \rmk}=\bm t
\right\}\\
&=\dot{\bigcup}_{\substack{\bm a\in \mathcal{FB}_{{\mfp_1,\mfp_2,N}}^{\bm t}
,\\
 }
 }
\left\{
\bm k\in C_{{{\mathbb E\lmk \mathfrak S_{\mfp_1,\mfp_2,+1}^{(N)} \rmk}}}\mid\;\;
\Psi^{(x,y)}\lmk \bm k \rmk\vert_{\caE^{(4)}\lmk\mathfrak S_{\mfp_1,\mfp_2,+1}^{(N)}\rmk}=\bm a
\right\}\\
&=\dot\bigcup_{\bm u\in  G^{\lv{J}_{{\mfp_1,\mfp_2,N}}\rv}}
\left\{
\bm k\in C_{{{\mathbb E\lmk \mathfrak S_{\mfp_1,\mfp_2,+1}^{(N)} \rmk}}}\mid\;\;
\Psi^{(x,y)}\lmk \bm k \rmk\vert_{\caE^{(4)}\lmk\mathfrak S_{\mfp_1,\mfp_2,+1}^{(N)}\rmk}=
\Xi_{{\mfp_1,\mfp_2,N}}(\bm t, \bm u)
\right\}.
\end{split}
\end{align}
Here we used Lemma \ref{risu}. By Lemma \ref{alr},
there exists $1-|G|^{\lv \caE_{\mathfrak S_{\mfp_1,\mfp_2,+1}^{(N)}}^{(1)}\rv}
$-correspondence
between $\mathcal{FB}_{{\mfp_1,\mfp_2,N}}$ and
$C_{{\mathbb E\lmk \mathfrak S_{\mfp_1,\mfp_2,+1}^{(N)} \rmk}}$.
Combining this fact and (\ref{alpaka}), for any $\bm t\in \mathcal{PT}_{{\mfp_1,\mfp_2,N}}$,
we have
\begin{align}
\begin{split}
&\lv
\left\{
\bm k\in C_{{{\mathbb E\lmk \mathfrak S_{\mfp_1,\mfp_2,+1}^{(N)} \rmk}}}\mid
\Psi^{(x,y)}\lmk \bm k \rmk\vert_{\partial \mathbb V\lmk \mathfrak S_{\mfp_1,\mfp_2,+1}^{(N)} \rmk}=\bm t
\right\}
\rv\\
&=\sum_{\bm u\in  G^{\lv{J}_{{\mfp_1,\mfp_2,N}}\rv}}
\lv
\left\{
\bm k\in C_{{{\mathbb E\lmk \mathfrak S_{\mfp_1,\mfp_2,+1}^{(N)} \rmk}}}\mid\;\;
\Psi^{(x,y)}\lmk \bm k \rmk\vert_{\caE^{(4)}\lmk\mathfrak S_{\mfp_1,\mfp_2,+1}^{(N)}\rmk}
=\Xi_{{\mfp_1,\mfp_2,N}}(\bm t, \bm u)
\right\}
\rv\\
&=\sum_{\bm u\in  G^{\lv{J}_{{\mfp_1,\mfp_2,N}}\rv}}
|G|^{\lv\caE_{\mathfrak S_{\mfp_1,\mfp_2,+1}^{(N)}}^{(1)}\rv}
=\lv G\rv^{\lv {J}_{{\mfp_1,\mfp_2},{N} }\rv}
|G|^{\lv\caE_{\mathfrak S_{\mfp_1,\mfp_2,+1}^{(N)}}^{(1)}\rv}.
\end{split}
\end{align}
\end{proof}

\section{Restriction of $\omega_0$ to a cone shape area}\label{res}
We use the notation from section \ref{intro}, \ref{coffe}, \ref{acls}
but we do not assume
$G$ to be abelian.

Let $\Gamma$ be a convex cone in $\bbR^2$ with apex $\bm a_\Gamma\in [0,1]\times[0,1]$.
Because of the $\bbZ^2$-translation invariance of the model, it suffices to consider this case.
The boundary of $\Gamma$ consists of two lines
$L_1^\Gamma:=\bm a_\Gamma+\bbR_{\ge 0}\bm e_{\theta_1}$,
$L_2^\Gamma:=\bm a_\Gamma+\bbR_{\ge 0}\bm e_{\theta_2}$,
with $0<\theta_2-\theta_1\le\pi$, because
$\Gamma$ is convex.
Here $\bm e_\theta:=(\cos\theta,\sin\theta)$.
We concentrate on the case $0\le \theta_1<\frac\pi 2$.
The proof of other cases are the same, just $\frac\pi 2,\pi,\frac{3\pi}{2}$-rotate
the following argument.

Note that three cases can occur:
\begin{description}
\item[(1)] $0\le \theta_1<\theta_2\le\frac\pi 2$,
\item[(2)] $0\le\theta_1<\frac\pi 2\le \theta_2\le \pi$,
\item[(3)] $0\le\theta_1<\frac\pi 2<\pi<\theta_2\le\pi+\theta_1<\frac{3\pi}2$.
\end{description}
For each case, we fix $n_0,m_0\in\bbN$ so that
\begin{description}
\item[(1)]$\frac{m_0}{n_0}< \tan \theta_2$,
\item[(2)]$n_0=m_0=1$,
\item[(3)]$\frac{m_0}{n_0}< \tan \theta_2$.
\end{description}
Then for case (1),(2) $\Gamma$ is in the upper-half plane and there exists $M_0,N_0\in\bbN$ such that
\begin{align}
\begin{split}
3 \le \lv
\left\{
x\in\bbZ\mid
S_{(x,M)}\subset \Gamma\cap \Lambda_N^{(n_0,m_0)}
\right\}
\rv
\end{split}
\end{align}
for all $M,N\in\bbN$ with
$M_0\le M\le N$, $N_0\le N$.
For case (3), 
there exists $N_0\in\bbN$ such that
\begin{align}
\begin{split}
3 \le \lv
\left\{
x\in\bbZ\mid
S_{(x,M)}\subset \Gamma\cap \Lambda_N^{(n_0,m_0)}
\right\}
\rv
\end{split}
\end{align}
for all $M,N\in\bbN$ with
$-Nm_0\le M\le Nm_0$, $N_0\le N$.
For (3) case, we set $M_0:=-\infty$.
For $M_0\in\bbZ\cup\{-\infty\}$, set
\begin{align}
H^U_{M_0}:=
\left\{(x,y)\in\bbZ^2\mid
y\ge M_0
\right\}.
\end{align}
Consider the set of squares
inside of
\begin{align}
\begin{split}
\Gamma\cap\bbZ^2\cap H^U_{M_0}.
\end{split}
\end{align}
By considering the boundary of the set given by these squares,
for each case (1)-(3),
we obtain a point 
$v_0=(x,y)\in\bbZ^2$, and 
$\mfp_1,\mfp_2\in \mathfrak P(x,y)$
a pair of paths well-separated with respect to $N_0,n_0, m_0$ satisfying
(1)-(3) in Setting \ref{kame}, respectively.
We use the notation from section \ref{kuma} for this 
$v_0=(x,y)\in\bbZ^2$, and 
$\mfp_1,\mfp_2\in \mathfrak P(x,y)$ freely.

Recall that $\caP_{\lmk \Gamma\cap\bbZ^2\cap H^U_{M_0}\rmk}$
denotes all the squares in $\lmk \Gamma\cap\bbZ^2\cap H^U_{M_0}\rmk$
and $\widetilde{\lmk \Gamma\cap\bbZ^2\cap H^U_{M_0}\rmk}$
denotes the edges inside of these squares.
We set 
\begin{align}
\begin{split}
&\hat\Gamma:=\widetilde{\lmk \Gamma\cap\bbZ^2\cap H^U_{M_0}\rmk},\\
&{\hat\Gamma_N}:=\widetilde{\lmk \Gamma\cap\Lambda_N^{(n_0,m_0)}\cap H^U_{M_0}\rmk}=\mathbb E\lmk \mathfrak S_{\mfp_1,\mfp_2,+1}^{(N)} \rmk.
\end{split}
\end{align}

In this section, we obtain an explicit formula for the restriction of $\omega_0=\varphi_0$ to 
 $\caB_{\hat\Gamma_N}$.
\begin{lem}\label{gred}
With the notation above, for any $N\ge N_0$,
 there is some number $c_{\Gamma,N}>0$ such that
\begin{align}
\begin{split}
&\varphi_0\lmk
\ket{\bm h}_{{\hat\Gamma_N}}\bra{\bm k}
\rmk\\
&=c_{\Gamma,N}\cdot 
 \chi\lmk
\bm k_{}, 
\bm h_{}\in C_{\sppe}
\rmk\chi\lmk
\Psi^{v_0}(\bm h_{})\vert_{\partial {\mathbb V\lmk \mathfrak S_{\mfp_1,\mfp_2,+1}^{(N)} \rmk}}=
\Psi^{v_0}(\bm k_{})\vert_{\partial {\mathbb V\lmk \mathfrak S_{\mfp_1,\mfp_2,+1}^{(N)} \rmk}}
\rmk\\
\end{split}
\end{align}
for any $\bm h, \bm k\in G^{\hat\Gamma_N}=G^{ {\sppe}}$.
\end{lem}
\begin{rem}
Note because $\mathfrak S_{\mfp_1,\mfp_2,+1}^{(N)} $ satisfies condition S,
$\Psi^{v_0}(\bm h_{})$ is well-defined for
any $\bm h_{}\in C_{\sppe}$.
\end{rem}
\begin{proof}
By the definition of $\varphi_0$, for any
$\bm h, \bm k\in G^{{\hat\Gamma_N}}$, 
\begin{align}\label{tamago}
\begin{split}
&\varphi_{0}\lmk
\ket{\bm h}_{{\hat\Gamma_N}}\bra{\bm k}
\rmk\\
&=
\sum_{\bm m\in G^{\mathfrak {OE}_{\mfp_1,\mfp_2}^{(N)}}}
\sum_{\bm g\in G^{\mathfrak {BE}_{\mfp_1,\mfp_2}^{(N)}}}
\varphi_{0}\lmk
\ket{\bm h,\bm m,\bm g}_{\hat\Lambda_N^{(n_0,m_0)}}\bra{\bm k,\bm m,\bm g}
\rmk\\
&=\sum_{\bm m\in G^{\mathfrak {OE}_{\mfp_1,\mfp_2}^{(N)}}}
\sum_{\bm g\in G^{\mathfrak {BE}_{\mfp_1,\mfp_2}^{(N)}}}
\frac{1}{|C_{{\hat\Lambda_N^{(n_0,m_0)}}}|}
\chi\lmk {\lmk \bm k,\bm m,\bm g\rmk, 
\lmk  \bm h,\bm m,\bm g\rmk\in C_{{\hat\Lambda_N^{(n_0,m_0)}}}}\rmk
\chi\lmk\lmk \bm k,\bm m,\bm g\rmk_{\partial{{\hat\Lambda_N^{(n_0,m_0)}}} }
= \lmk \bm h,\bm m,\bm g\rmk_{\partial{{\hat\Lambda_N^{(n_0,m_0)}}}}\rmk\\
&=\sum_{\bm m\in C_{\mathfrak {OE}_{\mfp_1,\mfp_2}^{(N)}}}
\sum_{g\in G^{{e_0^{\mfp_1,\mfp_2, (N)}}}}
\sum_{\tilde {\bm g}\in G^{\mathfrak {BE}_{\mfp_1,\mfp_2}^{(N)}\setminus \{{e_0^{\mfp_1,\mfp_2, (N)}}\}}}
\frac{1}{|C_{{\hat\Lambda_N^{(n_0,m_0)}}}|}
\chi\lmk {\lmk \bm k,\bm m,g, \tilde{\bm g}, \rmk, 
\lmk  \bm h,\bm m,g, \tilde{\bm g}\rmk\in C_{{\hat\Lambda_N^{(n_0,m_0)}}}}\rmk\\
&\quad\quad\quad\quad\quad\quad\quad\quad\quad\quad\quad\quad\chi\lmk
 \lmk \bm k,\bm m,g,\tilde{\bm g}\rmk_{\partial{{\hat\Lambda_N^{(n_0,m_0)}}}}
=\lmk \bm h,\bm m,g, \tilde{\bm g}\rmk_{\partial{{\hat\Lambda_N^{(n_0,m_0)}}} }\rmk
 \chi\lmk
\bm k, 
\bm h\in C_{{\mathbb E\lmk \mathfrak S_{\mfp_1,\mfp_2,+1}^{(N)} \rmk}}
\rmk\\
&
\quad \quad\quad\quad\quad\quad\quad\quad\quad\quad\quad\quad\chi\lmk
\Psi^{(x,y)}(\bm h)\vert_{{{\partial {\mathbb V\lmk \mathfrak S_{\mfp_1,\mfp_2,+1}^{(N)} \rmk}}}}=
\Psi^{(x,y)}(\bm k)\vert_{{{\partial {\mathbb V\lmk \mathfrak S_{\mfp_1,\mfp_2,+1}^{(N)} \rmk}}}}
\rmk.
\end{split}
\end{align}
Here we used Lemma \ref{gorira} equivalence of (ii),(iii) for the last line.
Recall from Lemma \ref{gorira} that if $\bm k, 
\bm h\in C_{{\mathbb E\lmk \mathfrak S_{\mfp_1,\mfp_2,+1}^{(N)} \rmk}}$,
$\Psi^{(x,y)}(\bm h)\vert_{{{\partial {\mathbb V\lmk \mathfrak S_{\mfp_1,\mfp_2,+1}^{(N)} \rmk}}}}=
\Psi^{(x,y)}(\bm k)\vert_{{{\partial {\mathbb V\lmk \mathfrak S_{\mfp_1,\mfp_2,+1}^{(N)} \rmk}}}}$,
for any
$\bm m\in C_{{\mathfrak {OE}_{\mfp_1,\mfp_2}^{(N)}}}$ and 
$g\in G^{{e_0^{\mfp_1,\mfp_2, (N)}}}$,
there exists a unique $\tilde {\bm g}\in G^{\mathfrak {BE}_{\mfp_1,\mfp_2}^{(N)}\setminus \{{e_0^{\mfp_1,\mfp_2, (N)}}\}}$
satisfying 
\begin{align}
\begin{split}
\lmk \bm k,\bm m,g, \tilde{\bm g}, \rmk, 
\lmk  \bm h,\bm m,g, \tilde{\bm g}\rmk
\in C_{{\hat \Lambda_N^{(n_0,m_0)}}}
\end{split}
\end{align}
and
\begin{align}
\begin{split}
 \lmk \bm k,\bm m,g,\tilde{\bm g}\rmk_{\partial{{\hat \Lambda_N^{(n_0,m_0)}}}}
= \lmk \bm h,\bm m,g, \tilde{\bm g}\rmk_{\partial{{\hat \Lambda_N^{(n_0,m_0)}}}}.
\end{split}
\end{align}
We have
\begin{align}
\begin{split}
(\ref{tamago})=&
\sum_{\bm m\in C_{\mathfrak {OE}_{\mfp_1,\mfp_2}^{(N)}}}
\sum_{g\in G^{{e_0^{\mfp_1,\mfp_2, (N)}}}}
\frac{1}{|C_{{\hat\Lambda_N^{(n_0,m_0)}}}|}
 \chi\lmk
\bm k, 
\bm h\in C_{{\mathbb E\lmk \mathfrak S_{\mfp_1,\mfp_2,+1}^{(N)} \rmk}}
\rmk\\
&\chi\lmk
\Psi^{(x,y)}(\bm h)\vert_{{{\partial {\mathbb V\lmk \mathfrak S_{\mfp_1,\mfp_2,+1}^{(N)} \rmk}}}}=
\Psi^{(x,y)}(\bm k)\vert_{{{\partial {\mathbb V\lmk \mathfrak S_{\mfp_1,\mfp_2,+1}^{(N)} \rmk}}}}
\rmk\\
=&\frac{\lv C_{\mathfrak {OE}_{\mfp_1,\mfp_2}^{(N)}}\rv\lv G\rv}{{|C_{{\hat\Lambda_N^{(n_0,m_0)}}}|}}
 \chi\lmk
\bm k, 
\bm h\in C_{{\mathbb E\lmk \mathfrak S_{\mfp_1,\mfp_2,+1}^{(N)} \rmk}}
\rmk\\
&\chi\lmk
\Psi^{(x,y)}(\bm h)\vert_{{{\partial {\mathbb V\lmk \mathfrak S_{\mfp_1,\mfp_2,+1}^{(N)} \rmk}}}}=
\Psi^{(x,y)}(\bm k)\vert_{{{\partial {\mathbb V\lmk \mathfrak S_{\mfp_1,\mfp_2,+1}^{(N)} \rmk}}}}
\rmk\\
\end{split}
\end{align}
\end{proof}


\section{Tracial state}
We use the notation from section \ref{intro}, \ref{coffe}, \ref{acls}, \ref{kuma}
but we do not assume
$G$ to be abelian.

We first consider the same setting as in section \ref{res},
and use notations from section \ref{kuma}.
For each $\bm t\in \mathcal{PT}_{{\mfp_1,\mfp_2,N}}$, set
\begin{align}
\caJ_{\bm t}^{(\Gamma, N)}
:=\left\{
\bm k\in C_{{\mathbb E\lmk \mathfrak S_{\mfp_1,\mfp_2,+1}^{(N)} \rmk}}\mid
\Psi^{(x,y)}(\bm k)\vert_{{{\partial {\mathbb V\lmk \mathfrak S_{\mfp_1,\mfp_2,+1}^{(N)} \rmk}}}}=\bm t
\right\}.
\end{align}
By Lemma \ref{gcnt}, we have 
\begin{align}\label{jkazu}
\begin{split}
\lv \caJ_{\bm t}^{(\Gamma, N)}\rv=
\lv G\rv^{\lv {J}_{{\mfp_1,\mfp_2},{N} }\rv}
\lv \caE_{\mathfrak S_{\mfp_1,\mfp_2,+1}^{(N)}}^{(1)}\rv
\end{split}
\end{align}
for any $\bm t\in  \mathcal{PT}_{{\mfp_1,\mfp_2,N}}$.
For $\bm t\in \mathcal{PT}_{{\mfp_1,\mfp_2,N}}$,
we set
\begin{align}
&\psi_{\bm t}^{N}:= \sum_{\bm k\in \caJ_{\bm t}^{(\Gamma, N)}}\ket{\bm k}_{\hat \Gamma_N},
\end{align}
By (\ref{jkazu}),
\begin{align}
\begin{split}
&\hat\psi_{\bm t}^{N}:=
\frac{1}{\sqrt{
\lv G\rv^{\lv {J}_{{\mfp_1,\mfp_2},{N} }\rv}
\lv \caE_{\mathfrak S_{\mfp_1,\mfp_2,+1}^{(N)}}^{(1)}\rv}}\;
\psi_{\bm t}^{N},\quad\quad \bm t\in  \mathcal{PT}_{{\mfp_1,\mfp_2,N}}.
\end{split}
\end{align}
are unit vectors.
We set
\begin{align}
\begin{split}
Q_N:=\sum_{\bm t\in \mathcal{PT}_{{\mfp_1,\mfp_2,N}}}
\ket{\hat\psi_{\bm t}^{N}}\bra{\hat\psi_{\bm t}^{N}}.
\end{split}
\end{align}

\begin{lem}\label{trac}
For any $N\ge N_0$, $Q_N$ is  the support of $\varphi_0$ on $\caB_{{\hat\Gamma_N}}$ and 
$\varphi_0\vert_{Q_N\caB_{{\hat\Gamma_N}}Q_N}$ is a tracial state.
Namely, we have
\begin{align}
\begin{split}
\varphi_0\lmk Q_N A Q_N B Q_N\rmk
=\varphi_0\lmk Q_N B Q_N A Q_N\rmk,\quad A,B\in \caB_{{\hat\Gamma_N}}.
\end{split}
\end{align}
\end{lem}
\begin{proof}
From Lemma \ref{gred}, for any $\bm k,\bm h\in G^{\hat\Gamma_N}$, we have
\begin{align}
\begin{split}
&\varphi_0\lmk
\ket{\bm h}_{{\hat\Gamma_N}}\bra{\bm k}
\rmk\\
&=c_{\Gamma,N}\cdot
 \chi\lmk
\bm k, 
\bm h\in C_{{\mathbb E\lmk \mathfrak S_{\mfp_1,\mfp_2,+1}^{(N)} \rmk}}
\rmk\chi\lmk
\Psi^{(x,y)}(\bm h)\vert_{{{\partial {\mathbb V\lmk \mathfrak S_{\mfp_1,\mfp_2,+1}^{(N)} \rmk}}}}=
\Psi^{(x,y)}(\bm k)\vert_{{{\partial {\mathbb V\lmk \mathfrak S_{\mfp_1,\mfp_2,+1}^{(N)} \rmk}}}}
\rmk\\
&=c_{\Gamma,N}
\sum_{\bm t\in \mathcal{PT}_{{\mfp_1,\mfp_2,N}}}
\braket{\psi_{\bm t}^{N}}{ 
\lmk \ket{\bm h}_{_{{\hat\Gamma_N}}}\bra{\bm k}\rmk\psi_{\bm t}^{N}}\\
&=c_{\Gamma,N}
\lv G\rv^{\lv {J}_{{\mfp_1,\mfp_2},{N} }\rv}
\lv \caE_{\mathfrak S_{\mfp_1,\mfp_2,+1}^{(N)}}^{(1)}\rv
\sum_{\bm t\in \mathcal{PT}_{{\mfp_1,\mfp_2,N}}}
\braket{\hat\psi_{\bm t}^{N}}{
\lmk \ket{\bm h}_{{\hat\Gamma_N}}\bra{\bm k}\rmk
\hat\psi_{\bm t}^{N}}\\
&=c_{\Gamma,N}
\lv G\rv^{\lv {J}_{{\mfp_1,\mfp_2},{N} }\rv}
\lv \caE_{\mathfrak S_{\mfp_1,\mfp_2,+1}^{(N)}}^{(1)}\rv
\Tr Q_N\lmk \lmk \ket{\bm h}_{{\hat\Gamma_N}}\bra{\bm k}\rmk\rmk\\
&=c_N\Tr Q_N\lmk \lmk \ket{\bm h}_{{\hat\Gamma_N}}\bra{\bm k}\rmk\rmk,
\end{split}
\end{align}
with $c_N:=
c_{\Gamma,N}
\lv G\rv^{\lv {J}_{{\mfp_1,\mfp_2},{N} }\rv}
\lv \caE_{\mathfrak S_{\mfp_1,\mfp_2,+1}^{(N)}}^{(1)}\rv$.
We then have
\begin{align}
\begin{split}
&\varphi_0(B)=c_N\Tr Q_N B Q_N,
\end{split}
\end{align}
for all $B\in\caB_{{\hat\Gamma_N}}$.
Hence we have
\begin{align}
\begin{split}
\varphi_0\lmk Q_N A Q_N B Q_N\rmk
=c_N\Tr\lmk Q_N A Q_N B Q_N\rmk
=c_N\Tr\lmk Q_N B Q_N A Q_N\rmk
=\varphi_0\lmk Q_N B Q_N A Q_N\rmk
\end{split}
\end{align}
for all $A, B\in\caB_{{\hat\Gamma_N}}$.

\end{proof}

Because ${Q_N}$ is  the support of $\varphi_0$ on $\caB_{{\hat\Gamma_N}}$
and ${Q_{N+1}}$ is  the support of $\varphi_0$ on $\caB_{\hat\Gamma_{N+1}}$,
we have
$
1-{Q_N}\le 1-{Q_{N+1}}
$
hence
\begin{align}\label{qdec}
{Q_{N+1}}\le {Q_N}.
\end{align}
Let $(\caH_0,\pi_0,\Omega_0)$ be the GNS triple of $\varphi_0$.
Because of (\ref{qdec}), there is a projection $q\in \pi_0(\caB_{\hat\Gamma})''$
such that 
\begin{align}
s-\lim\pi_0\lmk{Q_N}\rmk=q.
\end{align}
Because
\begin{align}
\pi_0\lmk{Q_N}\rmk\Omega_0=\Omega_0,
\end{align}
we have 
\begin{align}
q\Omega_0=\Omega_0.
\end{align}
In particular, $q\neq 0$.
\begin{lem}
Let $\omega$ be a state on $q\pi_0(\caB_{\hat\Gamma})''q$
such that 
\begin{align}
\omega(x):=\braket{\Omega_0}{ x\Omega_0},\quad x\in q\pi_0(\caB_{\hat\Gamma})''q.
\end{align}
Then $\omega$ is a faithful tracial state on $q\pi_0(\caB_{\hat\Gamma})''q$.
\end{lem}
\begin{proof}
We first show for any $A,B\in \caB_{\hat\Gamma,\rm{loc}}:=\cup_{M\in\bbN} \caB_{\hat\Gamma_M,\rm{loc}}$
that
\begin{align}\label{lcv}
\begin{split}
\braket{\Omega_0}{q\pi_0(A)q\pi_0(B)q\Omega_0}
=\braket{\Omega_0}{q\pi_0(B)q\pi_0(A)q\Omega_0}.
\end{split}
\end{align}
From Lemma \ref{trac}, for any $A,B\in \caB_{\hat\Gamma,\rm{loc}}$, we have
\begin{align}
\begin{split}
&\braket{\Omega_0}{\pi_0\lmk A\rmk \pi_0\lmk  {Q_N} \rmk\pi_0\lmk B \rmk\Omega_0}=
\braket{\Omega_0}{\pi_0\lmk {Q_N} A {Q_N} B {Q_N}\rmk\Omega_0}\\
&=\varphi_0\lmk Q_N A Q_N B Q_N\rmk
=\varphi_0\lmk Q_N B Q_N A Q_N\rmk\\
&=\braket{\Omega_0}{\pi_0\lmk {Q_N} B {Q_N} A {Q_N}\rmk\Omega_0}
=\braket{\Omega_0}{\pi_0\lmk  B\rmk\pi_0\lmk {Q_N} \rmk\pi_0\lmk A \rmk\Omega_0}
\end{split}
\end{align}
for $N$ large enough.
Taking $N\to\infty$ limit, we obtain
\begin{align}
\begin{split}
\braket{\Omega_0}{\pi_0\lmk A\rmk  q\pi_0\lmk B \rmk\Omega_0}=
\braket{\Omega_0}{\pi_0\lmk  B\rmk q \pi_0\lmk A \rmk\Omega_0},
\end{split}
\end{align}
proving (\ref{lcv}).
From (\ref{lcv}), we have
\begin{align}
\begin{split}
\omega\lmk qxqyq\rmk=\braket{\Omega_0}{qxqyq\Omega_0}
=\braket{\Omega_0}{qyqxq\Omega_0}
=\omega\lmk qyqxq\rmk,
\end{split}
\end{align}
for any $x, y\in \pi_0\lmk  \caB_{\hat\Gamma}\rmk''$.
Hence $\omega$ is a tracial state on $q\pi_0( \caB_{\hat\Gamma})''q$.

In order to show that $\omega$ is faithful, suppose that
 $x\in q\pi_0( \caB_{\hat\Gamma})''q$
satisfies $\omega(x^*x)=0$.
Then for any $y\in q\pi_0( \caB_{\hat\Gamma})''q$, we have
\begin{align}
\begin{split}
\lV xy\Omega_0\rV^2=\omega\lmk y^*x^*x y\rmk
= \omega\lmk x^*x yy^*\rmk= \omega\lmk x yy^*x^*\rmk
\le \lV y \rV^2\omega\lmk x x^*\rmk
=\lV y \rV^2\omega\lmk x^* x\rmk=0.
\end{split}
\end{align}
This means
\begin{align}
xy\Omega_0=0,\quad \text{for all}\quad  y\in q\pi_0( \caB_{\hat\Gamma})''q.
\end{align}
Furthermore, because 
\begin{align}
x\in q\pi_0( \caB_{\hat\Gamma})''q\subset \pi_0( \caB_{\hat\Gamma})''
\subset  \pi_0(\caB_{\hat\bbZ^2\setminus \hat\Gamma})',\quad q\in \pi_0( \caB_{\hat\Gamma})''\subset  \pi_0(\caB_{\hat\bbZ^2\setminus \hat\Gamma})',
\end{align}
we have
\begin{align}
\begin{split}
&x\pi_0\lmk A\rmk\pi_0\lmk B\rmk\Omega_0
=\pi_0\lmk B\rmk x\pi_0\lmk A\rmk\Omega_0
=\pi_0\lmk B\rmk  xq\pi_0\lmk A\rmk q\Omega_0=0,\\
&A\in  \caB_{\hat\Gamma}, \quad B\in \caB_{\hat\bbZ^2\setminus \hat\Gamma}.
\end{split}
\end{align}
Hence we obtain $x=0$, proving the faithfulness.
\end{proof}
\begin{lem}
Let $\Gamma$ be a convex cone in $\bbR^2$ with apex $\bm a_\Gamma\in [0,1]\times[0,1]$
given by 
$L_1^\Gamma:=\bm a_\Gamma+\bbR_{\ge 0}\bm e_{\theta_1}$,
$L_2^\Gamma:=\bm a_\Gamma+\bbR_{\ge 0}\bm e_{\theta_2}$,
with $0\le \theta_1<\frac\pi 2$, $0<\theta_2-\theta_1<\pi $.
Then 
$\pi_0\lmk \caB_{\widetilde{\lmk \Gamma\cap \bbZ^2\rmk}}\rmk''$ 
is not type $III$.
\end{lem}
\begin{proof}
Let $q$ be the projection in $\pi_0( \caB_{\hat\Gamma})''$ given above.
Suppose $p\in \pi_0( \caB_{\hat\Gamma})''$ is a projection
such that $p\sim q$ in $\pi_0( \caB_{\hat\Gamma})''$
and $p\le q$. We have $p=v^*v$, $q=vv^*$ for some $v\in \pi_0( \caB_{\hat\Gamma})''$.
Because $p\le q$, $v=qvq\in q\pi_0( \caB_{\hat\Gamma})''q$.
We then have
\begin{align}
\begin{split}
\omega\lmk q\rmk=\omega(vv^*)=\omega(v^*v)=\omega(p)
\end{split}
\end{align}
and for $q-p\ge 0$ we have
$\omega(q-p)=0$.
By the faithfulness of $\omega$, we have $p=q$.
Hence $q$ is finite in $\pi_0( \caB_{\hat\Gamma})''$.
Because $\pi_0( \caB_{\hat\Gamma})''$ has a non-zero finite projection,
it is not type $III$.
Note the difference between $\hat\Gamma$ and 
$\widetilde{\lmk \Gamma\cap\bbZ^2\rmk}$
is at most finite.
Therefore, $\pi_0\lmk \caB_{\widetilde{\lmk \Gamma\cap \bbZ^2\rmk}}\rmk''$ 
is not type $III$ either.

\end{proof}
The $\frac\pi 2,\pi,\frac{3\pi}{2}$-rotation of 
the above argument
and $\bbZ^2$-translation invariance of the model
allow us to extend the result as follows.

\begin{lem}\label{nanoha}
Let $\Gamma$ be a convex cone in $\bbR^2$.
Then $\pi_0\lmk \caB_{\widetilde{\lmk \Gamma\cap \bbZ^2\rmk}}\rmk''$  
is not type $III$.
\end{lem}

\section{Proof of Theorem \ref{main}}
Now finally we assume 
$G$ to be abelian.

\begin{proofof}[Theorem \ref{main}]
From \cite{FN} Theorem 4.2, there is a nontrivial superselection sector for abelian quautum double model.
Furthermore, Haag duality holds for the abelian quantum double model Theorem 3.1 \cite{FN} .
Therefore, by Lemma 5.5 \cite{ogata2022derivation},  for any cone $\Gamma\subset \bbR^2$,
$\pi_0\lmk \caB_{\widetilde{\lmk \Gamma\cap \bbZ^2\rmk}}\rmk''$
is a type $II_\infty$ factor or a type $III$-factor.
For a convex $\Gamma$, Lemma \ref{nanoha} then implies
that $\pi_0\lmk \caB_{\widetilde{\lmk \Gamma\cap \bbZ^2\rmk}}\rmk''$ is a a type $II_\infty$ factor.
\end{proofof}
{\bf Acknowledgment.}\\
{The author would like to thank Pieter Naaijkens for kind discussion.
This work was supported by JSPS KAKENHI Grant Number 19K03534 and 22H01127.
It was also supported by JST CREST Grant Number JPMJCR19T2.
}

\bibliographystyle{alpha}
\bibliography{type}

\end{document}